\documentclass [12 pt, lettersize, onecolumn]{IEEEtran}
\usepackage{mathtools}
\usepackage{amsmath,amsfonts}
\usepackage{algorithmic}
\usepackage{algorithm}
\usepackage{array}
\usepackage[caption=false,font=normalsize,labelfont=sf,textfont=sf]{subfig}
\usepackage{textcomp}
\usepackage{stfloats}
\usepackage{url}
\usepackage{verbatim}
\usepackage{graphicx}
\usepackage{dsfont}
\usepackage{enumitem} 
\usepackage{amsthm}
\usepackage{marvosym}
\usepackage{setspace}
\usepackage[normalem]{ulem}
\usepackage[T1]{fontenc}
\usepackage{bm}
\usepackage{hyperref}
\usepackage{cite}
\doublespace
\allowdisplaybreaks

\newtheorem{lem}{Lemma}
\newtheorem{thm}{Theorem}
\newtheorem{prp}{Proposition}
\newtheorem{cor}{Corollary}
\newtheorem*{rem*}{Remark}

\newtheorem{defn}{Definition}

\newcommand{\R}{\mathbb{R}}
\newcommand{\E}{\mathbb{E}}
\newcommand{\N}{\mathbb{N}}
\newcommand{\Pb}{\mathbb{P}}

\newcommand{\C}{\mathbb{C}}

\newcommand{\indicator}{\mathds{1}}
\newcommand{\Laplace}{\mathcal{L}}
\newcommand{\mtxLaplace}{\boldsymbol{\mathcal{L}}}
\newcommand{\sinr}{{\sf SINR}}
\newcommand{\sir}{{\sf SIR}}
\newcommand{\blkdiag}{~{\rm blkdiag}}

\begin{document}
\bstctlcite{IEEEexample:BSTcontrol}

\title{A Matrix Exponential Generalization of the Laplace Transform of Poisson Shot Noise}
\author{Nicholas R. Olson and Jeffrey G. Andrews 
\thanks{N. R. Olson and J. G. Andrews are with 6G@UT and WNCG at The University of Texas at Austin, Austin, TX, USA (email: nolson@utexas.edu, jandrews@ece.utexas.edu). }
\thanks{Manuscript received June 7, 2024; revised September 18, 2024.}}

\maketitle

\begin{abstract}
We consider a generalization of the Laplace transform of Poisson shot noise defined as an integral transform with respect to a matrix exponential. We denote this as the {\em matrix Laplace transform} and establish that it is in general a matrix function extension of the scalar Laplace transform. We show that the matrix Laplace transform of Poisson shot noise admits an expression analogous to that implied by Campbell's theorem.
We demonstrate the utility of this generalization of Campbell's theorem in two important applications: the characterization of a Poisson shot noise process and the derivation of the complementary CDF (CCDF) and meta-distribution of signal-to-interference-and-noise (SINR) models in Poisson networks. In the former application, we demonstrate how the higher order moments of Poisson shot noise may be obtained directly from the elements of its matrix Laplace transform. We further show how the CCDF of this object may be bounded using a summation of the first row of its matrix Laplace transform. For the latter application, we show how the CCDF of SINR models with phase-type distributed desired signal power may be obtained via an expectation of the matrix Laplace transform of the interference and noise, analogous to the canonical case of SINR models with Rayleigh fading. Additionally, when the power of the desired signal is exponentially distributed, we establish that the meta-distribution may be obtained in terms of the limit of a sequence expressed in terms of the matrix Laplace transform of a related Poisson shot noise process.
\end{abstract}

\begin{IEEEkeywords}
Stochastic Geometry, Poisson Shot Noise, Laplace Transform, Matrix Functions, Phase-Type Distributions, Meta-Distribution.
\end{IEEEkeywords}


\section{Introduction}
\label{sec:intro} 

{S}{tochastic} geometry has proven to be a powerful mathematical framework for the analysis of wireless networks, and has been used in a wide variety of settings over the past few decades, \cite{Hmamouche21} \cite{Haenggi12} \cite{baccelli10}. The principal utility of this framework is that it allows for tractable models of the spatial locations of transmitters and receivers in a network to be obtained through the machinery of point processes \cite{BaccelliSG}. This enables the development of stochastic models for the received power of direct and interfering transmissions across the network, which can then be used to characterize a variety of metrics -- such as coverage probability, ergodic capacity, and interference statistics, among others \cite{Hmamouche21}. 

The vast majority of the metrics studied in such analyses are functionals of the interference observed by a receiver at a particular location in the network, \cite{Ganti12}, \cite{Ganti09}, \cite{Schilcher16}. In general, the interference as observed at an arbitrary receiver located at $y \in \R^{d}$  may be described in terms of a cumulative shot noise process \cite[Def. 2.4.1]{BaccelliSG} of the form
\begin{align}
I(y) = \int_{\R^d} H(x;y)\Phi(dx),	
\label{eq:gen_shot_noise}
\end{align}
where $\Phi$ is a point process on $\R^{d}$ and $H: \R^{d} \times \R^d \rightarrow \R$ is a non-negative random field modeling received powers between transmitters at points in $\Phi$ and the receiver at $y$. 
Analytical tractability of functionals of interference therefore hinges on tractable characterizations of $I(y)$. 

There has been a large body of work dedicated to methods to maintain tractability  with respect to functionals of $I(y)$, see e.g. \cite{Xiao21}, and we do not attempt to provide a comprehensive summary here. Rather, we consider the common and important case where $I(y)$ corresponds to a Poisson shot noise process. In this case, $\Phi$ is a Poison Point Process (PPP) and $H$ is an independent random field which is also independent of $\Phi$. 
We develop a generalization of Campbell's theorem for the Laplace functional of a PPP which enables an integral transform of $I(y)$ with respect to a matrix exponential to be obtained in integral closed form. 
We term  such a transform the {\em matrix Laplace Transform} of $I(y)$ due to its similarity to the Laplace-Stieltjes transform. As we will show in the example applications considered in Sec. \ref{sec:pois_sn} and Sec. \ref{sec:sinr_analysis}, this generalization of Campbell's theorem allows for significantly improved tractability in a variety of settings pertinent to the analysis of wireless networks.
Poisson shot noise processes have applications to contexts outside of wireless networks, for instance to the several applications detailed in \cite{Lowen90} and to channel models for optical communications in \cite{Chakraborty07}, and thus the generalization Campbell's theorem we obtain is of interest to these applications as well.

\subsection{Campbell's Theorem and Related Work}
One of the key tools in analyzing functionals of Poisson shot noise is Campbell's Theorem, a version of which we restate here for completeness \cite[Prop. 2.1.4]{BaccelliSG}.
\begin{thm}(Campbell's Theorem)
\\
\label{thm:campbells_thm}
Let $\Phi$ be a PPP on $\R^d$ with intensity measure, $\Lambda$, and let $f:\R^d \rightarrow \R$ be a non-negative, measurable function. Then, the Laplace functional of $\Phi$, $\Laplace_{\Phi}(f) = \E\left[\exp\left(-\int_{\R^d} f(x) \Phi(dx) \right) \right]$, may be expressed as
\begin{align}
	\Laplace_{\Phi}(f) = \exp\left(- \int_{\R^d}\left(1 - e^{-f(x)}\right) \Lambda(dx) \right).
\end{align}
	
\end{thm}

As a point of clarification on the terminology we employ in this paper, Campbell's Theorem is more commonly taken to refer to the theorem due to Campbell which allows for the mean of a sum of functions indexed by a point process to be computed as an integral with respect to the mean measure of the point process, also called Campbell's averaging formula \cite[Prop. 1.2.5]{BaccelliSG}. However, Theorem \ref{thm:campbells_thm} and variants thereof are also attributable to Campbell, and are likewise referred to as Campbell's Theorem -- e.g. in \cite[Thm. 4.6]{Haenggi12} and \cite[Thm. 3.2]{Kingman93}. We employ the latter meaning throughout this paper.

Of particular interest to us is the following corollary of Campbell's theorem, which enables the Laplace transform of Poisson shot noise to be obtained.
\begin{cor} (Corollary of Campbell's Theorem)
\label{thm:campbells}
\\
Let $I(y)$ be as defined in (\ref{eq:gen_shot_noise}) such that $\Phi$ is a PPP on $\R^d$ with intensity measure, $\Lambda$, and $H:\R^d \times \R^d \rightarrow \R$ is a non-negative independent random field independent of $\Phi$. Let the Laplace transform of $H(x;y)$ be denoted as $\Laplace_H(\cdot; x,y)$. Then, for $s \in \C$ such that the real part of $s$ is strictly positive, the Laplace transform of $I(y)$, $\Laplace_{I(y)}(s) = \E\left[ e^{ -s I(y)} \right]$, is
\begin{align}
	\Laplace_{I(y)}(s) = \exp\left( -\int_{\R^d}\left(1 - \mathcal{L}_H(s ;x,y)\right) \Lambda(dx)\right).
	\label{eq:campbells_thm}
\end{align}
\end{cor}
\begin{proof}
	The proof is omitted for brevity as the result is known. 
\end{proof}
The general formula for the Laplace transform of Poisson shot noise in Corollary \ref{thm:campbells} has enabled tractable analysis in the characterization of a large class of wireless networks when the locations of transmitters follow a PPP. Perhaps the most commonly studied metric is coverage probability, which is the CCDF of the SINR as experienced by the typical receiver \cite{Baccelli06}. Campbell's theorem was employed to characterize the coverage probability of wireless ad hoc networks \cite{Baccelli06}, and later cellular networks \cite{Andrews11}, and is nearly universally employed in prior work for the coverage probability analysis of wireless networks. In what is by now a common analytical approach, when one considers an SINR model with exponentially distributed desired signal power\ (as in the case of Rayleigh fading), it is possible to obtain the CCDF of the SINR with respect to a certain expectation of the Laplace transform of the interference and noise \cite[Thm. 1]{Andrews11}. 

Of course, the Rayleigh fading assumption is not always suitable in many scenarios of interest. In the coverage probability analysis of mmWave \cite{Haichao22, Bai15}, multi-antenna networks \cite{ Li14, Bacha17, Tanbourgi15, Louie11}, and satellite communications \cite{Okati20, Song23, Park24}, 
Nakagami fading is usually more pertinent. Rayleigh fading is a special case of Nakagami fading when the shape parameter is equal to one, and is generally only applicable to single input single output (SISO) channel models with rich scattering. When the fading parameters follow Nakagami distributions with an integer shape parameter, the standard approach is to express the CCDF of the SINR in terms of an expectation of a finite summation of the higher order derivatives of the Laplace transform of the interference and noise \cite[Appendix III]{Ganti09}. This complicates the expression for coverage probability in these contexts, and necessitates a means of attaining some level of tractability for the derivatives. 

The majority of the exact approaches employed in prior work are based in some way on Fa{\`a} di Bruno's formula, which provides an explicit expression for the higher order derivatives of composite functions. In particular, for a function of the form $e^{-g(s)}$ such as (\ref{eq:campbells_thm}) the $j^{th}$ derivative admits the expression \cite{Johnson02}
\begin{align}
\frac{\partial^j}{\partial s^j} e^{-g(s)} = \sum_{k = 0}^{j}  (-1)^{k}e^{-g(s)} B_{j,k} \left(g^{(1)}(s), g^{(2)}(s), \dots, g^{(j - k +1)}(s) \right),
\label{eq:FDB}
\end{align}
where $g^{(k)}$ denotes the $k^{th}$ derivative of $g$ and $B_{j,k}$ is the $j^{th}$ Bell polynomial of order $k$. These may be defined recursively as
\begin{align}
B_{j,k}(s_1, \dots, s_{j -k  1}) = \sum_{i = 1}^{j-k} {j - 1 \choose i} s_{i + 1} B_{j - i,k-1}(s_1, \dots, s_{j - i - k}) && B_{0,0} = 1, ~B_{m,0} = B_{0,n} = 0, ~m,n > 0.
\end{align}
Fa{\`a} di Bruno's formula becomes increasingly difficult to evaluate as the order of the derivative grows large due to the number of terms involved in the recursion. Alternatively, (\ref{eq:FDB}) may be stated directly, but the direct form involves a summation over a number of terms that grow super-linearly with the order of the derivative. These challenges can be alleviated, to a degree, through the usage of more general Taylor series expansions, as in \cite{Ganti12} and \cite{Louie11}; through more concise recursive methods, such as \cite{Li14, Haichao22}; or approximations based on Alzer's lemma \cite{Bai15}.

Campbell's theorem has also proved useful for the characterization of the CCDF of interference modeled as Poisson shot noise. This is of particular interest in the SIR coverage analysis of networks without randomness in the desired signal power and the characterization of interference from wireless networks more broadly. Additionally, as we will show in Sec. \ref{sec:sinr_analysis}, characterizing the CCDF of Poisson shot noise allows for one to characterize the meta-distribution \cite{Haenggi16} of SINR models with exponentially distributed desired signal power. Due to the fact that Campbell's theorem allows for the determination of the characteristic function of Poisson shot noise \cite[Thm. 4.6.]{Haenggi12}, the characteristic function may be used to determine its CCDF via Fourier inversion \cite{Haenggi09}. Unfortunately however, except for a few special cases, the inversion of the characteristic function does not admit a closed form expression. Consequently, one must resort to numerical methods which are often computationally intensive \cite{DiRenzo14}. In light of these challenges, other approaches have instead leveraged bounds based on Chernoff-type approximations \cite{Weber07}, or have employed higher order moments of the shot noise. Similar to the case of coverage probability with Nakagami fading, the higher order moments are typically obtained via the derivatives of the Laplace transform.

Before proceeding, we note that our paper is most closely related to \cite{Schilcher16}, wherein the authors generalize a version of the Campbell-Mecke theorem to facilitate the computation of a broad class of higher order sum-product functionals that includes sum-product functionals of Poisson shot noise as a special case. Our characterization of matrix Laplace transforms in Proposition \ref{prp:mtx_lt} implies that the matrix Laplace transform of Poisson shot noise essentially corresponds to a similarity transform applied to an upper triangular matrix with entries corresponding to the expectation of sum-product functionals of Poisson shot noise of the form 
\begin{align}
	\E\left[\left(\int_{\R^d} H(x;y)\Phi(dx)\right)^{k} \exp\left(-s\int_{\R^d} H(x;y)\Phi(dx)\right) \right] &&k \in \N, s \in \C.
	\label{eq:mtx_lt_element}
\end{align}
When $s$ is real and non-negative, such sum-product functionals are a special  case of the more general sum-product functionals considered in \cite{Schilcher16}, and may be characterized using the results therein. However, the key novelty of our generalization of Campbell's theorem is that we express the matrix Laplace transform directly in terms of the matrix exponential of a matrix integral analogous to Corollary \ref{thm:campbells}. Consequently, our result provides significantly improved tractability for the example applications we consider as it alleviates the need to explicitly evaluate the relatively complicated summations required to characterize expressions such as (\ref{eq:mtx_lt_element}) when they are evaluated individually. 


\subsection{Contributions and Summary}
Motivated by these challenges and gaps in the literature, in this paper we provide a characterization of the matrix Laplace transform of Poisson shot noise and demonstrate how this characterization can provide improved tractability in the settings detailed in the previous subsection.

The development of the matrix exponential generalization of the Laplace transform, termed the matrix Laplace transform, and its characterization for non-negative Poisson shot noise are presented in Sec. \ref{sec:analysis}. Our principle contributions in this section are as follows.
\begin{itemize}

\item We define the matrix Laplace transform for general random variables and characterize the conditions under which it exist. In the case of non-negative random variables, we establish that the matrix Laplace transform is in general the natural matrix function \cite{Higham08}  extension of the typical scalar Laplace transform. These results are summarized in Proposition \ref{prp:mtx_lt} in Sec. \ref{sec:analysis}.
\item Building upon the results in Proposition \ref{prp:mtx_lt}, we derive an expression for the matrix Laplace transform of general non-negative Poisson shot noise. In particular, we establish various conditions under which its matrix Laplace transform may be expressed in terms of a matrix exponential of a certain matrix valued integral in a manner analogous to Corollary \ref{thm:campbells}. These results form the basis for our generalization of Campbell's theorem and are summarized in Theorem \ref{thm:gen_campbells} in Sec. \ref{sec:analysis}. Although similar to Corollary \ref{thm:campbells}, the characterization of the matrix Laplace transform of Poisson shot noise in Theorem \ref{thm:gen_campbells} is not a trivial extension of Corollary \ref{thm:campbells}.
\end{itemize} 


As we show in the applications considered in Sec. \ref{sec:pois_sn} and Sec. \ref{sec:sinr_analysis}, the generalization of Campbell's theorem in Theorem \ref{thm:gen_campbells} enables improved tractability in the important challenges mentioned in the previous sub-section, in addition to other more general settings. 
In Sec. \ref{sec:pois_sn} we demonstrate how Theorem \ref{thm:gen_campbells} may be used to tractably characterize a linear combination of samples of a non-negative Poisson shot noise process. The contributions of our paper pertinent to this section are as follows.
\begin{itemize}
\item We establish that the higher order moments of the aforementioned object may be obtained directly from the elements of its matrix Laplace transform. This is summarized in Corollary \ref{cor:sn_moments} in Sec. \ref{sec:pois_sn}. The resulting expression obtained in Corollary \ref{cor:sn_moments} is significantly more concise than would be afforded by approaches based on differentiating the scalar Laplace transform. Moreover, the general approach to obtaining the higher order moments via the matrix Laplace transform is applicable to any random variable for which the higher order moments in question are finite.
\item We additionally demonstrate how Theorem \ref{thm:gen_campbells} may be used to characterize the CCDF of the aforementioned object in terms of a limit of a sequence of approximations of the CCDF which are expressed in terms of its matrix Laplace transform. Notably, this characterization is such that upper and lower bounds on the CCDF may be obtained in terms of any term in the sequence. This result is summarized in Corollary \ref{cor:sn_dist} in Sec. \ref{sec:pois_sn}. The overall approach to the characterization of the CCDF detailed in Sec. \ref{sec:pois_sn} is applicable to any non-negative random variable, and is established in Proposition \ref{prp:approx_bounds} in Sec. \ref{sec:pois_sn}. This general approach to characterizing the distribution of Poisson shot noise provides a potentially more tractable alternative to the common approach of characterizing the distribution via characteristic function inversion and additionally allows for a straightforward, novel means of obtaining upper and lower bounds on its distribution.
\end{itemize}

In Sec. \ref{sec:sinr_analysis}, we consider the application of Theorem \ref{thm:gen_campbells} to the coverage probability and meta-distribution analysis of the SINR of a potentially non-stationary downlink Poisson cellular network. The contributions of our paper pertinent to this section are as follows.
\begin{itemize}
\item We show how Theorem \ref{thm:gen_campbells} may be used to obtain tractable expressions for the coverage probability of a typical user in the network when the desired signal power follows a general phase-type distribution \cite{Bladt17}. This setting includes the cases of Nakagami fading with an integer shape parameter and Rayleigh fading as special cases. In particular, the  expression for coverage probability we obtain in this case is the same as the well established form encountered in networks with Rayleigh fading -- the only difference being that the scalar Laplace transform in the latter case is replaced with the matrix Laplace transform in the former case. This result is summarized in Corollary \ref{cor:sinr_cov_ph} in Sec. \ref{sec:sinr_analysis}. Corollary \ref{cor:sinr_cov_ph} allows for the various methods developed for the analysis of coverage probability in Poisson networks with Rayleigh fading to be applied to the more general setting where the desired signal power follows a phase-type distribution. Additionally, it provides a more tractable expression for coverage probability than the commonly employed methods based on differentiating the scalar Laplace transform of the interference and noise terms of the SINR model, and holds for the entire class of phase-type random variables.


\item Finally, we additionally consider the application of Theorem \ref{thm:gen_campbells} to the analysis of the meta-distribution of the SINR of a typical user when the desired signal power is exponentially distributed. Namely, we show that the meta-distribution may be equivalently expressed in terms of the CCDF of a related Poisson shot noise process, and thus that Corollary \ref{cor:sn_dist} may be used to provide an alternative means of characterizing the meta-distribution. Notably, this approach allows for the analysis of the meta-distribution in a wider class of SINR models, as opposed to only SIR models, and provides a potentially more tractable alternative to methods relying on the Gil-Paelez theorem as in prior work \cite{Haenggi16}. 
\end{itemize}


Finally, the longer proofs of some of our results are deferred to a series of appendices at the end of the paper. Intermediate lemmas necessary for the generalization of Campbell's theorem in Theorem \ref{thm:gen_campbells} are stated and proved in Appendix \ref{app:int_lemmas}. 
The proof of Theorem \ref{thm:gen_campbells} is detailed in Appendix \ref{app:gen_campbells}. The proof of the approximation guarantees for our proposed method of approximating CCDFs using the matrix Laplace transform is stated in Appendix \ref{app:approx_bounds_proof}.

\subsection{Notation}

Throughout the paper we employ the following notation. Random variables are denoted as capital letters, e.g. $X$. For a random variable, $X$, we denote its Laplace transform as $\Laplace_X(s) = \E\left[ e^{-s X} \right]$. We denote the Erlang distribution with shape parameter $k$ and rate $\lambda$ as Erlang$(k, \lambda)$. 
The exponential distribution with rate $\lambda$ is denoted as Exp$(\lambda)$. We abbreviate {\em almost surely} as a.s. For a random field $H:\R^d \rightarrow \R$ and measure $\Lambda$ on $(\R^d, \mathcal{B}(\R^d))$, $H$ is said to be {\em independent $\Lambda$-$a.s.$} if there exists a set $B \in \mathcal{B}(\R^d)$ such that $(H(x))_{x \in B}$ is an independent random field and $\Lambda(\R^d \setminus B) = 0$.

Vectors are denoted in lowercase bold font, e.g. $\mathbf{x}$, and matrices are denoted in upper case bold font, e.g. $\mathbf{X}$. 
For $i \in \{1, \dots, n\}$, we take $\mathbf{e}_i$ to denote the $i^{th}$ standard basis vector in $\R^{n}$. $\mathbf{1}_{m \times n}$ and $\mathbf{0}_{m \times n}$ denote the  matrices of all ones or all zeros in $\R^{m \times n}$. We take $\mathbf{I}$ to denote the identity matrix in $\R^{n \times n}$. By$\blkdiag(\mathbf{A}_i)_{i = 1}^{k}$ we denote the block diagonal matrix composed by sub-matrices $(\mathbf{A}_i)_{i = 1}^{k}$. We denote the $i,j^{th}$ entry of a matrix, $\mathbf{A}$, as $\mathbf{A}_{i,j}$. We use slice notation to indicate specific rows or columns of $\mathbf{A}$, for instance $\mathbf{A}_{i, :}$ denotes the $i^{th}$ row. For a matrix, $\mathbf{A}$, $\exp(\mathbf{A})$ denotes the matrix exponential of $\mathbf{A}$. By the term {\em Jordan block matrix} we refer to the individual sub-matrices that in general compose a Jordan matrix. That is, matrices which have a single value, $\lambda \in \C$, along the diagonal entries, ones along the first upper off-diagonal entries, and zeros elsewhere. 

For an arbitrary function, $f(x)$, we denote the $j^{th}$ derivative of $f$ with respect to $x$ as $f^{(j)}$. Finally for $z \in \C$, $\Re(z)$ denotes the real part of $z$ and $\Im(z)$ denotes the imaginary part of $z$. We denote the non-negative real numbers as $\R_{+}$. For $y \in \R^d$ and $r \in \R_+$, $B(y, r)$ denotes the set $\{x \in \R^d : \lvert \lvert x - y \rvert \rvert \le r\}$.

\section{Matrix Laplace Transforms and Generalization of Campbell's Theorem}
\label{sec:analysis}
In this section, we establish the generalization of Campbell's theorem for the characterization of the matrix Laplace transform of Poisson shot noise. To that end, we first present the basic definitions and results necessary to characterize the notion of matrix Laplace transforms. Following this, the generalization of Campbell's theorem is presented in Theorem \ref{thm:gen_campbells}. Useful corollaries of Theorem \ref{thm:gen_campbells} are then presented: Corollaries \ref{cor:gen_campbells} and \ref{cor:alt_gen_campbells}. Corollaries \ref{cor:gen_campbells} and \ref{cor:alt_gen_campbells} are direct generalizations of Corollary \ref{thm:campbells} -- the only difference in the resulting expressions is that the scalar integral in Corollary \ref{thm:campbells} is replaced by an analogous matrix integral and matrix function in Corollaries \ref{cor:gen_campbells} and \ref{cor:alt_gen_campbells}, respectively. 

\subsection{Definition and Characterization of Matrix Laplace Transforms}

We define the {\em matrix Laplace transform} of a random variable, $X$, as the expectation of the matrix exponential of an arbitrary matrix, $\mathbf{S}$, scaled by $X$. We formalize this notion in the following definition.

\begin{defn} (Matrix Laplace Transform)
\label{def:mtx_lt}
\\
Let $X$ be a random variable and $\mathbf{S} \in \C^{n \times n}$ be a deterministic matrix. The matrix Laplace transform of $X$, $\mtxLaplace_X(\mathbf{S})$, is
\begin{align}
	\mtxLaplace_X(\mathbf{S}) = \E\left[ \exp\left(- \mathbf{S} X\right)\right].
\end{align}
\end{defn}

Whether or not the matrix Laplace transform of a random variable exists depends on the properties of the eigenvalues of the matrix, $\mathbf{S}$. We formalize this condition for general matrix Laplace transforms in Proposition \ref{prp:mtx_lt}. 
Our characterization leverages the notion of matrix functions as detailed in \cite{Higham08}. Hence, we first summarize the necessary definitions to formalize this concept before proceeding with the statement of Proposition \ref{prp:mtx_lt}.  

It will be useful to exploit the Jordan normal form of the matrix $\mathbf{S}$, which we denote as follows.
\begin{defn} (Jordan Normal Form)
\label{def:jordan}
\\
$\mathbf{S}$ may be expressed in terms of a block diagonal matrix $\mathbf{J}$
\begin{align}
	\mathbf{S} &= \mathbf{P} \mathbf{J} \mathbf{P}^{-1}, \nonumber
	\\
	&= \mathbf{P} \blkdiag\left(\mathbf{J}_i\right)_{i = 1}^{k} \mathbf{P}^{-1}
\end{align}
where $\mathbf{P}$ is a non-singular matrix, $k$ is an integer such that $k \le n$, and $(\mathbf{J}_i)_{i = 1}^{k}$ are Jordan block matrices with eigenvalue, $\lambda_i$, of size $m_i \times m_i$ such that $\sum_{i = 1}^k m_i = n$. We refer to $\mathbf{P}$ as the Jordan basis of $\mathbf{S}$, $\mathbf{J} = (\mathbf{J}_i)_{i = 1}^{k}$ as the Jordan normal form of $\mathbf{S}$, and $\mathbf{J}_i$ as the $i^{th}$ Jordan block matrix of $\mathbf{S}$.
\end{defn}


The matrix exponential is a matrix function in the sense of \cite[Def. 1.1 and 1.2]{Higham08}. We shall exploit properties of general matrix functions in our characterization of the matrix Laplace transform, which we summarize in the following definition.
\begin{defn} (Matrix Function)
\\
\label{def:mtx_function}
Let $f:\C \rightarrow \C$ be an arbitrary function. Let $\mathbf{S} \in \C^{n \times n}$ have $p$ distinct eigenvalues and let $(n_m)_{m = 1}^{p}$ denote the sizes of the largest Jordan block corresponding to each eigenvalue. 
\begin{enumerate}[label=\roman*)]
	\item $f$ is said to be defined on the spectra of $\mathbf{S}$ if $f^{(j)}(\lambda_m)$ exists $\forall j \in \{0, \dots, n_m\}, \forall m \in \{1, \dots, p\}$.
	\item Let $f$ be defined on the spectra of $\mathbf{S}$. Then the matrix function corresponding to $f$, $\boldsymbol{f}$, takes values in $\C^{n \times n}$ and is defined as
	\begin{align}
		\boldsymbol{f}(\mathbf{S}) 
		&= \mathbf{P} \blkdiag\left(\boldsymbol{f}(\mathbf{J}_i)\right)_{i = 1}^{k} \mathbf{P}^{-1}
	\end{align}
	where $\left(\boldsymbol{f}(\mathbf{J}_i)\right)_{i = 1}^{k}$ are upper triangular Toeplitz matrices with first row 
	\begin{align}
		\boldsymbol{f}(\mathbf{J}_i)_{1,:} =\begin{bmatrix}
			f(\lambda_i) &f^{(1)}(\lambda_i) &\frac{1}{2!}f^{(2)}(\lambda_i) &\dots &\frac{1}{(m_i - 1)!}f^{(m_i - 1)}(\lambda_i).
		\end{bmatrix}
	\end{align}
\end{enumerate}
	
\end{defn}

With these concepts defined, we are now ready to state our characterization of matrix Laplace transforms.

\begin{prp}(Characterization of Matrix Laplace Transforms)
\label{prp:mtx_lt}
\\
Let $X$ be a random variable taking values in $\C$ with matrix Laplace transform $\mtxLaplace_X$ and scalar Laplace transform $\Laplace_X$. Let $\mathbf{S} \in \C^{n \times n}$. The following hold.
\begin{enumerate}[label=\roman*)]
	\item $\mtxLaplace_X(\mathbf{S})$ exists iff $\mtxLaplace_X$ is convergent on the spectra of $\mathbf{S}$. This is, if
	\begin{align}
		\E\left[(-X)^{j-1}e^{-\lambda_i X} \right] \in \C && \forall j \in  \{1, \dots, m_i\}, ~\forall i \in \{1, \dots, k\}.
		\label{eq:mtx_lt_jordan_def}
	\end{align}
	If $\mtxLaplace_X(\mathbf{S})$ is convergent on the spectra of $\mathbf{S}$, it is the matrix function defined as
	\begin{align}
		\mtxLaplace_X(\mathbf{S}) = \mathbf{P}\blkdiag\left(\mtxLaplace_X(\mathbf{J}_i)\right)_{i = 1}^{k} \mathbf{P}^{-1}.
		\label{eq:gen_mtx_laplce}
	\end{align}
	where $\left(\mtxLaplace_X(\mathbf{J}_i)\right)_{i = 1}^{k}$ are upper triangular Toeplitz matrices with first row given by 
	\begin{align}
		\mtxLaplace_X(\mathbf{J}_i)_{1,j} = \frac{1}{(j-1)!}\E\left[(-X)^{j-1}e^{-\lambda_i X} \right] &&j \in \forall \{1, \dots, m_i\}.
		\label{eq:mtx_lt_hot_def}
	\end{align}
	\item Assume further that $X$ takes values only in $\R_+$, and that $\mathbf{S}$ is such that its eigenvalues lie in the region of convergence of $\Laplace_X$, and all eigenvalues corresponding to non-scalar Jordan blocks lie in the interior of the region of convergence of $\Laplace_X$. Then, $\Laplace_X$ is defined on the spectra of $\mathbf{S}$ and $\mtxLaplace_X(\mathbf{S})$ is the matrix function corresponding to $\Laplace_X$.
\end{enumerate}
\end{prp}
\begin{proof} We first prove i). Note that the matrix exponential is a matrix function defined for all $\mathbf{S} \in \C^{n \times n}$ \cite[Thm. 4.7]{Higham08}. Hence, it follows that
\begin{align}
	\mtxLaplace_X(\mathbf{S}) = \E\left[ \mathbf{P} \blkdiag\left(\exp(\mathbf{J}_i X)\right)_{i = 1}^{k} \mathbf{P}^{-1} \right].
\end{align}
When $\mtxLaplace_X(\mathbf{S})$ exists, by linearity of matrix multiplication it follows that
\begin{align}
	\mathbf{P}^{-1} \mtxLaplace_X(\mathbf{S}) \mathbf{P} &= \blkdiag\left(\E\left[\exp(\mathbf{J}_iX)\right]\right)_{i = 1}^{k} \nonumber
	\\
	&= \blkdiag\left(\mtxLaplace_X(\mathbf{J}_i)\right)_{i = 1}^{k},
\end{align}
in which case $\mtxLaplace_X$ is convergent on the spectra of $\mathbf{S}$. On the other hand, when $\mtxLaplace_X$ is convergent on the spectra of $\mathbf{S}$ it again follows by linearity that
\begin{align}
\mathbf{P} \blkdiag\left(\mtxLaplace_X(\mathbf{J}_i)\right)_{i = 1}^{k} \mathbf{P}^{-1} = \mtxLaplace_X(\mathbf{S}).
\end{align} 
Thus, $\mtxLaplace_X(\mathbf{S})$ exists.

	Part ii) is a direct consequence of \cite[Cor. 2]{Widder29}, in light of i). \cite[Cor. 2]{Widder29} states that, for all $\lambda$ in the interior of the region of convergence of $\Laplace_X$,
\begin{align}
\mathcal{L}^{(j)}_X(\lambda) = \E\left[(-X)^{j-1} e^{-\lambda X} \right] \in 
\C && \forall j \in \N
\label{eq:diff_exchange}
\end{align}
This follows by definition for $j = 0$ when $\lambda$ is in the region of convergence of $\Laplace_X$. Hence, so long as $\mathbf{S}$ satisfies the conditions stated in part ii), we have that $\Laplace_X$ is defined on the spectra of $\mathbf{S}$. Moreover, (\ref{eq:diff_exchange}) implies that $\mtxLaplace_X$ is convergent on the spectra of $\mathbf{S}$. Therefore it follows that $\mtxLaplace_X(\mathbf{S})$ is the matrix function corresponding to $\Laplace_X$.
\end{proof}

Part i) of Proposition \ref{prp:mtx_lt} indicates that the matrix Laplace transform of a random variable is simply the matrix function generalization of its scalar Laplace transform. The matrix Laplace transform always entails the evaluation of the scalar Laplace transform on the eigenvalues of $\mathbf{S}$, and when evaluated at non-diagonalizable matrices, the matrix Laplace transform entails the evaluation of the higher order terms of the form stated in (\ref{eq:mtx_lt_hot_def}). Under appropriate conditions on $\lambda$, as detailed for non-negative $X$ in part ii), these higher order terms are the higher order derivatives of the scalar Laplace transform. Hence, the matrix Laplace transform may be seen as an alternative method of obtaining the higher order derives of scalar Laplace transform. This implies its utility in the characterization of the higher order moments and distributional characteristics of $X$, as we demonstrate for the particular case of Poisson shot noise in Sec. \ref{sec:pois_sn}.
Overall, Proposition \ref{prp:mtx_lt} is important for the generalization of Campbell's theorem in Theorem \ref{thm:gen_campbells}, as it ensures that the matrix Laplace transform of a non-negative random variable always exists when $\mathbf{S}$ has eigenvalues with positive real part. Additionally, Proposition \ref{prp:mtx_lt} allows us to leverage established properties of matrix functions, as developed in \cite{Higham08}.

\subsection{Matrix Laplace Transform of Poisson Shot Noise}

Having established a well defined notion of matrix Laplace transforms, we now turn to our main result generalizing Campbell's theorem for the matrix Laplace transform of Poisson shot noise. In some sense, part ii) of Proposition \ref{prp:mtx_lt} implies that the corollary of Campbell's theorem stated in Corollary \ref{thm:campbells} is all that is required to obtain a characterization the matrix Laplace transform. However, this characterization would only hold for matrices with eigenvalues strictly within the region of convergence of the scalar Laplace transform and would further require the derivatives of the scalar Laplace transform to be obtained. Such a characterization would provide no further tractability in practical applications, in comparison to prior work, if used directly. 

Thus, alternatively, we establish an equivalent characterization of the matrix Laplace transform of Poisson shot noise in terms of an integral closed form matrix function that is analogous to Corollary \ref{thm:campbells}. As we will discuss further in Sec. \ref{sec:pois_sn} and Sec. \ref{sec:sinr_analysis}, this formulation provides improved tractability in a variety of practical applications. Moreover, since our proof is direct, it allows for us to consider the characterization of the matrix Laplace transform of Poisson shot noise for matrices that would not necessarily allow for the application of part ii) of Proposition \ref{prp:mtx_lt} to Corollary \ref{thm:campbells}. We first establish the generalization of Campbell's theorem for different types of Jordan block matrices in the following theorem. Note that, without loss of generality, we drop the dependence of the shot noise $I(y)$ presented in (\ref{eq:gen_shot_noise}) on $y$.

\begin{thm}(Matrix Laplace Transform of Poisson Shot Noise)
\label{thm:gen_campbells}
\\
Let $\Phi$ be a PPP on $\R^{d}$ with intensity measure $\Lambda$ and $H: \R^{d} \rightarrow \R_{+}$ be a random field independent of $\Phi$ such that $H$ is independent $\Lambda$-a.s.  Let the matrix Laplace transform of $H(x)$ be denoted as $\boldsymbol{\mathcal{L}}_H(\cdot;x)$. 

Consider the Poisson shot noise $I = \sum_{X_k \in \Phi} H(X_k)$.
The following hold.
\begin{enumerate}[label=\roman*)]
	\item Let $\mathbf{J} \in \mathbb{C}^{n \times n}$ be a Jordan block matrix with eigenvalue $\lambda \in \C$ such that $\Re(\lambda) > 0$.
Then the matrix Laplace transform of $I$ evaluated at $\mathbf{J}$, $\boldsymbol{\mathcal{L}}_I(\mathbf{J})$, exists and may be expressed as
\begin{align}
	\boldsymbol{\mathcal{L}}_I(\mathbf{J}) = \begin{cases}
		\exp\left( -\int_{\R^{d}}\left( \mathbf{I} - \boldsymbol{\mathcal{L}}_H(\mathbf{J};x) \right) \Lambda(dx) \right) &\left \lvert \int_{\R^d}\left(1 - \Laplace_H\left(\lambda;x\right)\right) \Lambda(dx) \right \rvert < \infty
		\\
		\mathbf{0}_{n \times n} &{\rm otherwise}.
	\end{cases}
	\label{eq:gen_campbells_thm}
\end{align}
\item Let $\mathbf{J}_0 \in \R^{n \times n}$, $n > 1$, be a Jordan block matrix with eigenvalue equal to zero. Then, if \\$\int_{\R^{d}}E[H(x)^{n-1}]\Lambda(dx)$ converges,
\begin{align}
	\boldsymbol{\mathcal{L}}_I(\mathbf{J}_0) = \exp\left( -\int_{\R^{d}}\left( \mathbf{I} - \boldsymbol{\mathcal{L}}_H(\mathbf{J}_0;x) \right) \Lambda(dx) \right).
\end{align}
Otherwise, $\boldsymbol{\mathcal{L}}_I(\mathbf{J}_0)$ does not exist.

\item Let $\mathbf{J} \in \C^{n \times n}$ be a Jordan block matrix with eigenvalue $\lambda \in \C$ such that $\Re(\lambda) = 0$. Assume that $\int_{\R^d}\left(1 - \Laplace_H(s;x) \right) \Lambda(dx)$ is convergent for all $s \in \R_{+}$ and that $\int_{\R^{d}}E[H(x)^{n-1}]\Lambda(dx)$ is convergent. Then,
\begin{align}
	\boldsymbol{\mathcal{L}}_I(\mathbf{J}) = 
		\exp\left( -\int_{\R^{d}}\left( \mathbf{I} - \boldsymbol{\mathcal{L}}_H(\mathbf{J};x) \right) \Lambda(dx) \right).
	\label{eq:char_gen_campbells}
\end{align}

\end{enumerate}

\end{thm}

\begin{proof}
	See Appendix \ref{app:int_lemmas} and Appendix \ref{app:gen_campbells}.
	\end{proof}
	
	Theorem \ref{thm:gen_campbells}, in tandem with part i) of Proposition \ref{prp:mtx_lt}, leads to a complete characterization of the matrix Laplace transform of Poisson shot noise which analogous to Corollary \ref{thm:campbells}. 
\ifdefined\usePiii
	Before stating this result, it will be useful to define the following extension of the matrix exponential.
	
	\begin{defn} (Extended Matrix Exponential)\\
	\label{def:ext_expm}
	The matrix function, $\overline{\exp}$, is the extension of the matrix exponential to upper triangular matrices with potentially infinite entries along the diagonal. In particular,
	 let $\mathbf{A} = \mathbf{D} + \mathbf{N}$, where $\mathbf{N} \in \C^{n \times n}$ is an upper triangular nilpotent matrix and $\mathbf{D}$ is a $ n \times n$ diagonal matrix with entries in $\bar{\C} = \{z: s.t.~ \Re(z) \in \R \cup \{-\infty\}, \Im(z) \in \R\}$, then
	\begin{align}
		\overline{\exp}(\mathbf{A})= \blkdiag(\exp(\mathbf{D}_{i,i}))_{i = 1}^{n} \exp(\mathbf{N}).
	\end{align}
	\end{defn}
	
	Using this definition, we state our complete generalization of Campbells Theorem for the matrix Laplace transform of a Poisson shot noise process as the following corollary.
	\else
	This is formalized in the following corollary of Theorem \ref{thm:gen_campbells} for matrices with eigenvalues in the interior of the right half complex plane.
	\fi
	
	\begin{cor}(Generalization of Campbell's Theorem.)
	\label{cor:gen_campbells}
	
	\ifdefined\usePiii
\begin{enumerate}[label=\roman*)]
\item 
\fi
Let $\mathbf{S} \in \C^{n \times n}$ be a matrix such that the real parts of its eigenvalues are strictly positive. Assume that $\int_{\R^d}\left(1 - \Laplace_H(s;x) \right) \Lambda(dx)$ is convergent for all $s \in \R_+$. Then, $I$ is finite a.s. and
\begin{align}
		\mtxLaplace_I(\mathbf{S}) = \exp\left( -\int_{\R^{d}}\left( \mathbf{I} - \boldsymbol{\mathcal{L}}_H(\mathbf{S};x) \right) \Lambda(dx) \right).
		\label{eq:comp_gen_campbells}
	\end{align}
	Otherwise, $I$ is infinite a.s. and $\mtxLaplace_I(\mathbf{S}) = \mathbf{0}_{n \times n}$.
	
\ifdefined\usePiii
\item 	 Let $\mathbf{S} \in \C^{n \times n}$ be a matrix with non-negative eigenvalues. Let $\mathbf{J}$ denote the Jordan normal form of $\mathbf{S}$ with corresponding Jordan basis, $\mathbf{P}$. Let $q$ denote the dimension of the largest Jordan block of $\mathbf{S}$ corresponding an eigenvalue with real part equal to zero. If $q \le 1$ assume that $\int_{\R^d}\left(1 - \Laplace_H(s;x) \right) \Lambda(dx) < \infty$ for all $s \in \R_{+}$. Otherwise, if $q > 1$ assume that $\int_{\R^{d}}E[H(x)]^{q-1}\Lambda(dx) < \infty$. Then, the matrix Laplace transform of $I$ exists and is given as
	\begin{align}
		\mtxLaplace_I(\mathbf{S}) = \mathbf{P} ~\overline{\exp}\left( -\int_{\R^{d}}\left( \mathbf{I} - \boldsymbol{\mathcal{L}}_H(\mathbf{J};x) \right) \Lambda(dx) \right) \mathbf{P}^{-1}.
	\end{align}
	When $\int_{\R^{d}} \left( 1- \Laplace_H(s;x) \right) \Lambda(dx)$ is convergent on the eigenvalues of $\mathbf{S}$, then
	\begin{align}
		\mtxLaplace_I(\mathbf{S}) = \exp\left( -\int_{\R^{d}}\left( \mathbf{I} - \boldsymbol{\mathcal{L}}_H(\mathbf{S};x) \right) \Lambda(dx) \right).
	\end{align}
	
\end{enumerate}	
\fi
	\end{cor}
	\begin{proof}
	\ifdefined\usePiii
	Part ii) is immediate from Theorem \ref{thm:gen_campbells}, Definition \ref{def:ext_expm}, and Proposition \ref{prp:mtx_lt}. 
	
	To establish part i), it
	\else
	Note that the fact that $I$ is finite a.s. when $\int_{\R^d}\left(1 - \Laplace_H(s;x) \right) \Lambda(dx)$ is convergent for all \\ $s \in \R_+$ and infinite a.s. otherwise follows from Lemma \ref{lem:finiteness} in Appendix \ref{app:int_lemmas}. 
	
	First assume that $I$ is finite a.s. It
	\fi
suffices to show that $\Laplace_I(s) \ne 0$, for $s \in \C$ such that $\Re(s) > 0$ when $\int_{\R^d}\left(1 - \Laplace_H(s;x) \right) \Lambda(dx)$ is convergent for all $s \in \R_{+}$. Assume that this is not the case, and let $\lambda \in \C$, $\Re(\lambda) > 0$, be such a that $\Laplace_I(\lambda) = 0$. In light of Corollary \ref{thm:campbells}, this occurs only if $\int_{\R^d}\left(1 - \Laplace_H\left(\lambda;x\right)\right) \Lambda(dx)$ is divergent. Hence, by Theorem \ref{thm:gen_campbells} and Proposition \ref{prp:mtx_lt}, we must have that
	\begin{align}
	 \Laplace_I^{(j)}(\lambda) = 0 && j \in \N.
	\end{align}
Thus, $\Laplace_I(s)$ is a flat function at $s = \lambda$. Moreover, since $\lambda$ is in the interior of the region of convergence of $\Laplace_I$, by \cite[Cor. 2]{Widder29} we have that $\Laplace_I$ is holomorphic. Hence, it follows that $\Laplace_I$ is analytic in the neighborhood of $\lambda$. Consequently, we must have that $\Laplace_I$ is 0 in the neighborhood of $\lambda$. This is a contradiction as this, in turn, implies that $\Laplace_I(s) = 0$ for all $s \in \C$ such that $\Re(s) > 0$. This cannot be the case as the fact that $\int_{\R^d}\left(1 - \Laplace_H(s;x) \right) \Lambda(dx)$ is convergent for all $s \in \R_+$ implies that $\Laplace_I(s) > 0$ for $s \in \R_+$ in light of Corollary \ref{thm:campbells_thm}. 

The fact that $\Laplace_I(s) \ne 0$, for $s \in \C$ such that $\Re(s) > 0$, implies that $\int_{\R^{d}}\left( \mathbf{I} - \boldsymbol{\mathcal{L}}_H(\mathbf{J};x) \right) \Lambda(dx)$ is always well defined when $\int_{\R^d}\left(1 - \Laplace_H(s;x) \right) \Lambda(dx)$ is convergent for all $s \in \R_{+}$. Thus, (\ref{eq:comp_gen_campbells}) follows from Theorem \ref{thm:gen_campbells} and Proposition \ref{prp:mtx_lt}. The fact that $\mtxLaplace_I(\mathbf{S}) = \mathbf{0}_{n \times n}$ when  $I$ is infinite a.s. follows from Lemma \ref{lem:mtx_lt_int} in Appendix \ref{app:gen_campbells} and Proposition \ref{prp:mtx_lt}.
	\end{proof}

Hence, for all shot noise models of practical interest, (\ref{eq:comp_gen_campbells}) always holds. This formulation of the matrix Laplace transform of Poisson shot noise is more tractable than the direct application of Proposition \ref{prp:mtx_lt} as it alleviates the need to explicitly determine the derivatives of the corresponding scalar Laplace transform, $\mathcal{L}_I(s)$. One need only determine the matrix Laplace transforms of the random field at distinct indices, $\boldsymbol{\mathcal{L}}_H(\mathbf{S};x)$, which is usually more straightforward.

Note further that, while the integral in (\ref{eq:comp_gen_campbells}) involves the evaluation over all $n^2$ elements of $\boldsymbol{\mathcal{L}}_H(\mathbf{S};x)$, one need only consider an integral over $n$ total terms. As a consequence of Proposition \ref{prp:mtx_lt}, we may express the right hand side of (\ref{eq:comp_gen_campbells}) in terms of the Jordan normal form of $\mathbf{S}$ as
\begin{align}
	\exp\left( -\int_{\R^{d}}\left( \mathbf{I} - \boldsymbol{\mathcal{L}}_H(\mathbf{S};x) \right) \Lambda(dx) \right) = \mathbf{P} \blkdiag\left(\exp\left( -\int_{\R^{d}}\left( \mathbf{I} - \boldsymbol{\mathcal{L}}_H(\mathbf{J}_i;x) \right) \Lambda(dx) \right)\right)_{i = 1}^{k} \mathbf{P}^{-1}.
\end{align}
By definition, each of the component matrix Laplace transforms, $\boldsymbol{\mathcal{L}}_H(\mathbf{J}_i;x)$, are upper triangular Toeplitz matrices. Hence, we need only integrate the entries corresponding to their first row. More explicitly, if we define the function $T(\mathbf{r})$ to produce the upper triangular Toeplitz matrix with first row given by $\mathbf{r}$, we may express (\ref{eq:comp_gen_campbells}) as
\begin{align}
	\exp\left( -\int_{\R^{d}}\left( \mathbf{I} - \boldsymbol{\mathcal{L}}_H(\mathbf{S};x) \right) \Lambda(dx) \right) = \mathbf{P} {\blkdiag}\left(\exp\left( - T\left(\int_{\R^{d}}\left( \mathbf{e}_1^{\rm T} -  \boldsymbol{\mathcal{L}}_H(\mathbf{J}_i;x)_{1,:} \right) \Lambda(dx) \right)\right)\right)_{i = 1}^{k} \mathbf{P}^{-1}
\end{align}
This only requires at most $n$ integrals to be performed, which can be done in parallel when evaluated numerically.

Finally, note that (\ref{eq:comp_gen_campbells}) admits the equivalent form stated in the following corollary of Corollary \ref{cor:gen_campbells}.
\begin{cor} (Alternative Generalization of Campbell's Theorem)
\label{cor:alt_gen_campbells}
\\
For $s \in \C$ such that $\Re(s) > 0$, let $f(s)$ be defined as
\begin{align}
	f(s) = \int_{\R^d}\left(1 - \Laplace_H(s;x) \right) \Lambda(dx).
\end{align}	
Let $\mathbf{S}$ be a matrix such that the real parts of its eigenvalues are strictly positive. Then, when $I$ is finite a.s. the matrix function extension of $f$, $\boldsymbol{f}(\mathbf{S})$, exists and
\begin{align}
	\mtxLaplace_I(\mathbf{S}) = \exp(-\boldsymbol{f}(\mathbf{S})).
	\label{eq:alt_gen_campbells}
\end{align}
\end{cor}
\begin{proof}
For $\mathbf{S}$ as given in the statement of the corollary, we have that the eigenvalues of $\mathbf{S}$ are in the interior of the region of convergence of $\Laplace_I$.  Thus, using Proposition \ref{prp:mtx_lt} and Corollary \ref{thm:campbells}, we may determine $\mtxLaplace_I(\mathbf{S})$ in terms of the higher order derivatives of (\ref{eq:campbells_thm}) in Corollary \ref{thm:campbells}. Proposition \ref{prp:mtx_lt} further ensures that these derivatives are finite. Corollary \ref{cor:gen_campbells}, ensures that $f(s)$ is convergent for all $s \in \C$ such that $\Re(s) > 0$. Then appealing to Fa{\`a} di Bruno's formula in (\ref{eq:FDB}), finiteness of $f(s)$ implies that $f(s)$ must be defined on the spectra of $\mathbf{S}$. The expression for $\mtxLaplace_I(\mathbf{S})$ in (\ref{eq:alt_gen_campbells}) consequently follows from Proposition \ref{prp:mtx_lt}, Corollary \ref{thm:campbells}, and the composition theorem for matrix functions \cite[Thm. 1.18]{Higham08}.
\end{proof}

Consequently, in light of Definition \ref{def:mtx_function}, Corollary \ref{cor:alt_gen_campbells} implies that we  may equivalently characterize $\mtxLaplace_I(\mathbf{S})$ in terms of the integrals of the scalar Laplace transforms of $H$, $\left( \int_{\R^{d}}\left( 1 -  \Laplace_H(\lambda_i;x) \right) \Lambda(dx)\right)_{i = 1}^{k}$, and their higher order derivatives up to $(m_i)_{i = 1}^{k}$. This provides an alternate formulation for $\mtxLaplace_I(\mathbf{S})$ which alleviates the need to explicitly employ Fa{\`a} di Bruno's formula.


\section{Characterization of the Distribution and Moments of Poisson Shot Noise}
\label{sec:pois_sn}
We now consider several applications of the generalization of Campbell's theorem in Theorem \ref{thm:gen_campbells}, which demonstrate its utility to practical scenarios. In this section, we first show how the result may by used to tractably characterize linear combinations of samples of a general non-negative Poisson shot noise process. To that end, a characterization of the higher order moments is presented in Corollary \ref{cor:sn_moments} in Sec. \ref{subsec:shot_noise_char_moments}, and approximations and bounds for the CCDF of this object are presented in Corollary \ref{cor:sn_dist} in Sec. \ref{subsec:shot_noise_char_dist}. Approximation guarantees for our proposed approximation method, which is applicable to any non-negative random variable, are further established in Proposition \ref{prp:approx_bounds} in Sec. \ref{subsec:shot_noise_char_dist}.



\subsection{Phase-Type Distributions}
\label{subsec:phase_type}
The applications of Theorem \ref{thm:gen_campbells} considered in this section and Sec. \ref{sec:sinr_analysis}  make use of {\em phase-type distributions} in some manner, and we first summarize their key properties before proceeding. A phase-type random variable, $Z$, is characterized by two parameters: a matrix, $\mathbf{G} \in \R^{n \times n}$, called the sub-generator matrix, and a vector, $\mathbf{p} \in [0,1]^{n}$, called the sub-PMF. We concisely denote the distribution of $Z$ as $Z \sim {\rm PH}(\mathbf{G}, \mathbf{p})$. Phase-type distributions are a general class of non-negative random variables which correspond to the hitting time to the absorbing state of an $n+1$ state continuous time Markov chain (CTMC) with $n$ transient states and one absorbing state. Thus, $\mathbf{G}$ is the sub-matrix of the transition rate matrix corresponding to the transient states (termed phases), and $\mathbf{p}$ is the subset of the initial distribution of the CTMC corresponding to the transient phases.

Of particular use to us in the following subsections and Sec. \ref{sec:sinr_analysis}, is that the CCDF of a phase-type distribution may be expressed in terms of a matrix exponential \cite[Thm. 3.1.7]{Bladt17}. Namely for $Z \sim {\rm PH}(\mathbf{G}, \mathbf{p})$ we have
\begin{align}
	\Pb(Z \ge \tau) = \mathbf{p}^{\rm T} \exp\left( \mathbf{G} \tau \right) \mathbf{1}_{N \times 1}.
	\label{eq:phase_type_cdf}
\end{align}
Note that the negation of feasible sub-generator matrices, i.e. $-\mathbf{G}$, have eigenvalues with strictly positive real parts \cite[Cor. 3.1.15]{Bladt17}, and, thus, allow for the application of Corollaries \ref{cor:gen_campbells} and \ref{cor:alt_gen_campbells}. In light of this, in contrast to the typical convention, we shall represent phase-type distributions in terms of the negation of their sub-generator matrices, $\mathbf{S} = -\mathbf{G}$.

The class of phase-type distributions is quite broad and includes the exponential and Erlang distributions as special cases. In particular, an Erlang$(n, \lambda)$ distribution corresponds to a PH$(-\lambda \mathbf{Q}, \mathbf{e}_1)$ phase type model, where $\mathbf{Q} \in \R^{n \times n}$ is an upper triangular Toeplitz matrix with first row $\left(\mathbf{Q}\right)_{1,:} = [	1 ~-1 ~\mathbf{0}_{1 \times n - 2}]$. Thus, in wireless applications with Rayleigh and Nakagami fading with an integer shape parameter, the corresponding fading powers belong to the class of phase type distributions -- the exponential and Erlang distributions, respectively. Moreover, phase-type distributions are dense with respect to random variables with support on $\R_+$ \cite[3.2.10]{Bladt17}. Consequently, they may be used to represent any non-negative random variable to arbitrary precision.

\subsection{Formulation of Shot Noise Model}
\label{subsec:shot_noise_char_mod}
In both of the following sub-sections, we characterize the higher order moments and CCDF of the following shot noise model. Namely, consider the general Poisson shot noise process $I(y)$ as defined in (\ref{eq:gen_shot_noise}) and define
\begin{align}
	\tilde{I} = \sum_{i = 1}^{m} \alpha_i I(y_i) && \alpha_i \ge 0, y_i \in \R^{d}, y_i \ne y_j.
	\label{eq:int_model_v2}
\end{align}
The model we consider for $\tilde{I}$ may be taken to correspond to a space-time interference model as discussed in \cite{Haenggi09}, and allows one to capture spatial correlations of interference arising from a Poisson shot noise process.
It will be useful to observe that $\tilde{I}$ may be expressed as
\begin{align}
\tilde{I} &= \int_{\R^d} \sum_{i = 1}^{m} \alpha_i H(x; y_i) \Phi(dx) \nonumber
\\
&= \int_{\R^d} \tilde{H}(x) \Phi(dx),
\end{align}
where we have defined $\tilde{H}(x) = \sum_{i = 1}^{m} \alpha_i H(x;y_i)$. In light of this formulation, for $m > 1$ we additionally stipulate that $\Lambda$ is diffuse. This ensures that $\tilde{H}$ is independent $\Lambda$-a.s. Hence $\tilde{I}$ is itself an instance Poisson shot noise. Note that for $m =1$ we do not require $\Lambda$ to be diffuse, and $\tilde{I}$ is simply an arbitrary instance of general non-negative Poisson shot noise.

\subsection{Characterization of Higher Order Moments}
\label{subsec:shot_noise_char_moments}

We may directly obtain the higher order moments of $\tilde{I}$, in addition to a necessary and sufficient condition for the higher order moments of $\tilde{I}$ to be finite, using the following corollary of Theorem \ref{thm:gen_campbells}.

\begin{cor}	(Moments of Poisson Shot Noise)
\\
\label{cor:sn_moments}
Let $\tilde{I}$ be Poisson shot noise of the form stated in (\ref{eq:int_model_v2}), and $n \in \N$. Let $\mathbf{J}_0 \in \R^{n \times n}$ be a Jordan block matrix with eigenvalue equal to zero. Then, $\E[\tilde{I}^{n-1}]$ is finite iff
$\int_{\R^{d}}\E[\tilde{H}(x)^{n-1}] \Lambda(dx)$ is convergent, in which case
\begin{align}
	\E[\tilde{I}^{n-1}] = (-1)^{n-1}(n-1)!\left(\exp\left( -\int_{\R^{d}}\left( \mathbf{I} - \prod_{i = 1}^{m}\boldsymbol{\mathcal{L}}_H\left(\alpha_i\mathbf{J}_0 ; x - y_i \right)\right) \Lambda(dx) \right) \right)_{1,n}.
	\label{eq:shot_noise_moments}
\end{align}
\end{cor}
\begin{proof}
This corollary exploits the fact that, for an arbitrary non-negative random variable, $X$, in light of Proposition \ref{prp:mtx_lt}
\begin{align}
	\mtxLaplace_X(\mathbf{J}_0)_{1, j} = \frac{(-1)^{j-1}}{(j-1)!} \E\left[  X^{j-1}\right] && j \in \{1, \dots, n\}.
	\label{eq:mtx_lt_moments}
\end{align}
The corollary then follows from (\ref{eq:mtx_lt_moments}) and Theorem \ref{thm:gen_campbells}.
\end{proof}

The fact that $\mathbf{J}_0$ is a Jordan block matrix simplifies the form of $\mtxLaplace_{\tilde{H}}(\mathbf{J}_0; x)$. In particular, the integral in (\ref{eq:shot_noise_moments}) is simply
\begin{align}
\left(\int_{\R^{d}}\left( \mathbf{I} - \prod_{i = 1}^{m}\boldsymbol{\mathcal{L}}_H\left(\alpha_i\mathbf{J}_0 ; x - y_i \right)\right) \Lambda(dx) \right)_{1,j} = 
\begin{cases}
\frac{(-1)^{j}}{(j-1)!}\int_{\R^{d}}\E\left[ \tilde{H}(x)^{j-1} \right] \Lambda(dx). & j >1
\\
0 & j = 1
\end{cases}
\end{align}
Thus, the $n^{th}$ moment of $\tilde{I}$ may be obtained as a function of integrals of the first $n$ moments of $\tilde{H}(x)$ with respect to the $\Lambda$. The resulting expression for $\E[\tilde{I}^{n-1}]$ is notably concise and relatively tractable compared to methods based on directly evaluating the higher order derivative of the scalar Laplace transform.

\subsection{Characterization of the CCDF}
\label{subsec:shot_noise_char_dist}
We may further exploit Theorem \ref{thm:gen_campbells} to characterize approximations of and bounds on the CCDF of $\tilde{I}$, $\bar{F}_{\tilde{I}}(\tau)$. Except for a few special cases, the CCDF of $\tilde{I}$ is often not available in closed form. Consequently, we consider characterizing the CCDF as
\begin{align}
	\bar{F}_{\tilde{I}}(\tau) = \lim_{N \to \infty} \Pb\left(\tilde{I} \ge \tau  Z_N \right),
	\label{eq:gen_aprx}
\end{align}
for an appropriately chosen sequence of non-negative random variables, $(Z_N)_{N \in \N}$, which converge to $1$ a.s. and are independent of $\tilde{I}$. This type of approximation has the desirable property that each term in the limit in (\ref{eq:gen_aprx}) corresponds to a feasible CCDF
\begin{align}
	P^{(
N)}_{\tilde{I}}(\tau) = \Pb\left(\tilde{I} \ge \tau  Z_N \right).
\end{align}
The rationale for this approach is that, if we can construct $(Z_N)_{N \in \N}$ such that $P^{(
N)}_{\tilde{I}}(\tau)$ is tractable, for large enough $N$, $P^{(
N)}_{\tilde{I}}(\tau)$ will be a good approximation of $\bar{F}_{\tilde{I}}(\tau)$ which is easier to obtain.

Given that their CDFs may be expressed in terms of a matrix exponential and their generality, we consider candidates for $(Z_N)_{N \in \N}$,  from the family of phase-type distributions. In particular, we construct $(Z_N)_{N \in \N}$ such that
\begin{enumerate}
	\item For each $N \in \N$, $Z_N$ is an $N$-phase PH$(-\mathbf{S}_N, \mathbf{p}_N)$ random variable.
	\item $\E[Z_N]$ = 1.
	\item $\lim_{N \to \infty} Z_N = 1$ a.s.
\end{enumerate} 
Under this general construction, we have that $P^{(N)}_{\tilde{I}}(\tau)$ may be characterized in terms of the matrix Laplace transform of $\tilde{I}$:
\begin{align}
	{P}_{\tilde{I}}^{(N)}(\tau) &= \E\left[\Pb\left(\tilde{I} \ge \tau  Z_N \lvert \tilde{I}\right) \right] \nonumber
	\\
	&= 1 - \mathbf{p}_N^{\rm T} \E\left[ \exp\left( - \tau^{-1} \mathbf{S}_N  \tilde{I}\right)\right] \mathbf{1}_{N \times 1} \nonumber
	\\
	&= 1 - \mathbf{p}_N^{\rm T} \boldsymbol{\mathcal{L}}_{\tilde{I}} \left( \tau^{-1} \mathbf{S}_N\right)\mathbf{1}_{N \times 1}.
	\label{eq:approx_ccdf_def}
\end{align}

To make full use of this approach to characterize $\bar{F}_{\tilde{I}}(\tau)$, it remains to determine the parameters for the phase-type distributions for each term of $(Z_N)_{N \in \N}$. To that end, we establish the following proposition, Proposition \ref{prp:approx_bounds}, which indicates that taking  $Z_N \sim {\rm Erlang}(N,N)$ for all $N \in \N$ is optimal in terms of minimizing an upper bound on the approximation error of ${P}_{\tilde{I}}^{(N)}(\tau)$ with respect to $\bar{F}_{\tilde{I}}(\tau)$. Indeed, more broadly, Proposition \ref{prp:approx_bounds} establishes that this choice of $(Z_N)_{N \in \N}$ is optimal for minimizing an upper bound on the approximation error associated with the application of the method detailed above to the characterization of the CCDF of any arbitrary non-negative random variable. The proposition additionally establishes a method of obtaining upper and lower bounds on the true CCDF in terms of the approximate CCDFs.

\begin{prp}(Approximation Guarantees for ${P}_I^{(N)}(\tau)$)
\label{prp:approx_bounds}
	\\
	Let $X$ be an arbitrary non-negative random variable with CCDF $\bar{F}_{X}(\tau)$, and $(\tilde{Z}_N)_{N \in \N}$ be any sequence of phase-type random variables independent of $X$ such that $\tilde{Z}_N$ has $N$ phases, $\E[\tilde{Z}_N] = 1$, and $(\tilde{Z}_N)_{N \in \N}$ converges to 1 a.s. Moreover, let $\tilde{P}_X^{(N)}(\tau)= \Pb\left({X} \ge \tau  \tilde{Z}_N \right)  $. Then, the following hold
	\begin{enumerate}[label=\roman*)]
		\item $\tilde{P}_X^{(N)}(\tau)$ converges to $\bar{F}_{X}(\tau)$ as $N \rightarrow \infty$.
		\item Define $\tilde{e}^{(N)}(\tau)$ as
		\begin{align}
			\tilde{e}^{(N)}(\tau) = \max\left(\E\left[ \left(\inf_{\theta \ge 0} \mathcal{L}_{\tilde{Z}_N}(-\theta)e^{-\theta X \tau^{-1}}\right) \indicator\{X \ge \tau\}\right], \E\left[ \left(\inf_{\theta \ge 0} \mathcal{L}_{\tilde{Z}_N}(\theta)e^{\theta X \tau^{-1}}\right) \indicator\{X < \tau\}\right]\right).
			\label{eq:error_def}
		\end{align}
		Then
		\begin{align}
		\left \lvert \bar{F}_{X}(\tau) - \tilde{P}_X^{(N)}(\tau)\right \rvert \le \tilde{e}^{(N)}(\tau).
		\label{eq:aprx_err_pii}
		\end{align}
		\item Let $(Z_N)_{N \in \N}$ be a sequence of Erlang$(N,N)$ random variables. Note that $(Z_N)_{N \in \N}$ is a feasible sequence of phase-type random variables with respect to the conditions on $(\tilde{Z}_N)_{N \in \N}$. Let ${P}_X^{(N)}(\tau)= \Pb\left({X} \ge \tau  {Z}_N\right) $ and let ${e}^{(N)}(\tau)$ be as defined in (\ref{eq:error_def}) using to $Z_N$.
		Then
		\begin{align}
			{e}^{(N)}(\tau) \le \tilde{e}^{(N)}(\tau).
			\label{eq:erlang_aprx_err}
		\end{align}
		Moreover, fix $\epsilon >0$ and let $\delta(\epsilon) >0$ satisfy $\delta(\epsilon) - \log(1 + \delta(\epsilon)) = \log\left( \epsilon^{-\frac{1}{N}}\right)$. Then,
		\begin{align}
			{P}_X^{(N)}\left( \frac{\tau}{1-\delta(\epsilon)} \right) - \epsilon \le \bar{F}_X(\tau) \le (1 - \epsilon)^{-1} {P}_X^{(N)}\left( \frac{\tau}{1+\delta(\epsilon)} \right).
			\label{eq:aprx_bounds}
		\end{align}
	\end{enumerate}
\end{prp}
\begin{proof}
	See appendix \ref{app:approx_bounds_proof}.
\end{proof}
Proposition \ref{prp:approx_bounds} implies that for moderately sized $N$, ${P}_{\tilde{I}}^{(N)}(\tau)$ is a reasonable approximation of $\bar{F}_{\tilde{I}}(\tau)$. For instance at $N = 100$ and $\epsilon = 10^{-4}$, we have $\delta(\epsilon) = 0.4927$. This corresponds to shifting ${P}_{\tilde{I}}^{(N)}(\tau)$ by roughly $2$ dB for the upper bound, and $3$ dB for the lower bound. Moreover, in light of part iii) it provides a general means for obtaining upper bound lower bounds on $\bar{F}_{\tilde{I}}(\tau)$ in terms of $P^{(N)}_{\tilde{I}}(\tau)$.

Proposition \ref{prp:approx_bounds} leads to the following approximation of $\bar{F}_{\tilde{I}}(\tau)$, which we state as the following corollary of Proposition \ref{prp:approx_bounds} and Theorem \ref{thm:gen_campbells}.

\begin{cor} (Approximations and Bounds for Shot Noise CCDFs)
\\
\label{cor:sn_dist}
Let ${\tilde{I}}$ be Poisson shot noise of the form stated in (\ref{eq:int_model_v2}), and assume that $\tilde{I}$ is finite a.s. The following hold.
\begin{enumerate}[label=\roman*)]
	\item Let $N \in \N$ denote the order of the Erlang$(N,N)$ model used in ${P}_{\tilde{I}}^{(N)}(\tau)$ and $\mathbf{Q}^{(N)} \in \R^{N \times N}$ be the negation of the sub-generator matrix corresponding to the Erlang$(N,1)$ distribution. Define the matrix
	\begin{align}
		\mathbf{S}_{\tilde{I}}^{(N)}(\tau) = \int_{\R^{d}}\left( \mathbf{I} - \prod_{i = 1}^{m} \mtxLaplace_{{H}}\left(\alpha_i \frac{N}{\tau}\mathbf{Q}^{(N)};x - y_i\right)\right) \Lambda(dx). 
	\end{align}
	The CCDF of $\tilde{I}$, $\bar{F}_{\tilde{I}}(\tau)$, is approximated as ${P}_{\tilde{I}}^{(N)}(\tau)$, where
\begin{align}
	{P}_{\tilde{I}}^{(N)}(\tau) = 1 - \mathbf{e}_1^{\rm T} \exp\left( -\mathbf{S}_{\tilde{I}}^{(N)}(\tau)\right)\mathbf{1}_{N \times 1},
	\label{eq:int_dist_aprx}
\end{align}
and
\begin{align}
	\bar{F}_{\tilde{I}}(\tau) = \lim_{N \to \infty} {P}_{\tilde{I}}^{(N)}(\tau)
\end{align}
\item Moreover, the $\bar{F}_{\tilde{I}}(\tau)$ admits upper and lower bounds of the form specified in (\ref{eq:aprx_bounds}) with respect to ${P}_{\tilde{I}}^{(N)}(\tau)$.
\end{enumerate}
\end{cor}
\begin{proof}
This corollary is immediate in light of Proposition \ref{prp:approx_bounds}  and Corollary \ref{cor:gen_campbells} as $\mathbf{Q}^{(N)}$ has positive real eigenvalues and $\tilde{I}$ is finite a.s.
\end{proof}

We may exploit the fact that $\mathbf{Q}^{(N)}$ is related to it Jordan normal form via a diagonal similarity transform (see the proof of Lemma \ref{lem:mtx_bounds} in Appendix \ref{app:int_lemmas}) to obtain
\begin{align}
	\mtxLaplace_{\tilde{H}}\left(\frac{N}{\tau}\mathbf{Q}^{(N)}; x\right)_{1,j} = \left(\frac{N}{\tau}\right)^{j-1} \frac{(-1)^{j-1}}{(j-1)!}\Laplace_{\tilde{H}}^{(j-1)}\left(\frac{N}{\tau}; x\right).
\end{align}
Hence, as mentioned in the end of Sec. \ref{sec:analysis}, in the evaluation of $\mathbf{S}_{\tilde{I}}^{(N)}(\tau)$, we need only consider the integral
\begin{align}
	\int_{\R^{d}} \left(\mathbf{e}_1 - \begin{bmatrix}
		\Laplace_{\tilde{H}}\left(\frac{N}{\tau};x\right) &\left(\frac{N}{\tau}\right)^{1}\frac{(-1)}{1!}\Laplace_{\tilde{H}}^{(1)}\left(\frac{N}{\tau};x\right) &\dots &\left(\frac{N}{\tau}\right)^{N-1}\frac{(-1)^{N-1}}{(N-1)!}\Laplace_{\tilde{H}}^{(N-1)}\left(\frac{N}{\tau};x\right)
	\end{bmatrix}\right) \Lambda(dx).
	\label{eq:cdf_aprx_int}
\end{align}
Equivalently, using Corollary \ref{cor:alt_gen_campbells} and Definition \ref{def:mtx_function} we may obtain $\mathbf{S}_{\tilde{I}}^{(N)}(\tau)$ in terms of the first $N$ derivatives of
\begin{align}
\int_{\R^{d}} \left(1 - 
		\Laplace_{\tilde{H}}\left(\frac{N}{\tau}s;x\right) \right) \Lambda(dx)
		\label{ee:scalr_cdf_int}
\end{align}
at $s = 1$.
Note that, interestingly, (\ref{eq:cdf_aprx_int}) implies that $\mathbf{S}_{\tilde{I}}^{(N)}(\tau)$ is a sub-generator matrix. Consequently, ${P}_{\tilde{I}}^{(N)}(\tau)$ corresponds to the probability that a PH$(-\mathbf{S}_{\tilde{I}}^{(N)}(\tau), \mathbf{e}_1)$ random variable is less than 1.

\section{SINR Analysis in Downlink Poisson Cellular Networks}
\label{sec:sinr_analysis}

In addition to providing tractable characterizations of Poisson shot noise, the generalization of Campbell's theorem in Theorem \ref{thm:gen_campbells} allows for improved tractability in the characterization of the CCDF and meta-distribution of SINR models associated with from Poisson networks. To that end, we consider the analysis of an exemplary SINR model pertinent to a downlink Poisson cellular network. In particular, when the desired signal power follows a phase-type distribution, the distribution of the SINR may be obtained in a manner similar to the classical form encountered in the case where the desired signal power is exponentially distributed, see e.g. \cite[Thm. 1]{Andrews11}, using the matrix Laplace transform of the interference. This result is stated in Corollary \ref{cor:sinr_cov_ph} in Sec. \ref{subsec:phase_type_cov}. 
Moreover, when the desired signal power follows an exponential distribution, the meta-distribution of the SINR may be expressed in terms of the distribution of a certain shot noise process which itself may be characterized using Corollary \ref{cor:sn_dist} and Proposition \ref{prp:approx_bounds}. This is summarized in Proposition \ref{prp:sinr_meta_dist} in Sec. \ref{subsec:sinr_meta}.

\subsection{Summary of SINR Model}

In both subsections, we consider the following SINR model which corresponds to a general downlink Poisson cellular network. Note, however, that the results obtained in this section may be replicated for other types of network models. Let $\Phi$ be a PPP on $\R^d$ with intensity measure $\Lambda$ denoting the locations of base stations (BSs), and consider a user receiver located at $y \in \R^d$. The user connects with the nearest BS, $X_y \in \Phi$, defined as
\begin{align}
	X_y = \inf\{ \lvert \lvert X - y \rvert \rvert : X \in \Phi \}.
\end{align} 
All other BSs in the network are interferers from the perspective of the user at $y$. Denote these points of $\Phi$ as $\Phi^{! X_y} = \Phi - \delta_{X_y}$. Each transmitted signal undergoes distance dependent path loss, $\ell : \R_+ \rightarrow \R_+$. 
The power of the signal transmitted by the serving BS (without accounting for path loss) is denoted as $P_y$, and is in general a random variable independent of $\Phi$. The signal powers of the interfering BSs (again prior to the application of path loss) arise from the independent random field $H: \R^d \times \R^d \rightarrow \R_+$ that is also independent of $P_y$ and $\Phi$. $\ell(r)$ is assumed only to be monotonically decreasing in $r$ and, similar to the path loss models considered in \cite{AlAmmouri19}, that
\begin{align}
	\int_{\R^d} \E\left[H(x;y) \right] \ell\left( \lvert \lvert x - y \rvert \rvert \right) \Lambda(dx) < \infty.
	\label{eq:pl_cond}
\end{align}
Consequently, the interference observed by the user at $y$ is 
\begin{align}
I(y) = \int_{\R^d \setminus B(y, \epsilon)} H(x;y) \ell \left( \lvert \lvert x - y \rvert \rvert \right) \Phi^{! X_y}(dx) && \epsilon > 0.
\end{align}
Note that (\ref{eq:pl_cond}) ensures that the $I(y)$ has a finite mean conditioned on $X_y$.
Putting these terms together and letting $L_y = \ell\left(\lvert \lvert X_y - y \rvert \rvert\right)$, the SINR observed by the user at $y$ admits the expression
\begin{align}
	\sinr(y) = \frac{P_y L_y}{I(y) + W},
	\label{eq:sinr_mod}
\end{align}
where $W$ is a constant denoting the noise power.

The performance of the cellular network as captured by this SINR model may be quantified with respect to two commonly studied metrics. The first is the coverage probability with respect to the typical user located at $y$, defined as
\begin{align}
P_{\rm cov}(\tau) = \Pb^y\left(\sinr(y) \ge \tau \right),
\end{align}
where $\Pb^y$ denotes the Palm distribution of the point process of users. As is typical, the user point process is assumed to be independent of the other random features of the network. Hence, the Palm distribution reduces to the nominal probability measure. The second metric of interest is the meta-distribution of the SINR with respect to the typical user at $y$. As first considered in \cite{Haenggi16}, the meta-distribution is defined with respect to the conditional success probability
\begin{align}
P_{\rm S}(\tau) = \Pb\left(\sinr(y) \ge \tau \lvert \Phi\right).
\end{align}
Then the meta-distribution is defined as
\begin{align}
	\bar{F}_{P_{\rm S}}(\zeta) = \Pb^y\left(P_{\rm S}(\tau) \ge \zeta \right) &&\zeta \in [0,1].
\end{align}
Under certain conditions pertaining to the distribution of $P_{y}$, the generalization of Campbell's theorem in Theorem \ref{thm:gen_campbells} provides improved tractability in the analysis of the coverage probability and meta-distribution.

\subsection{Coverage Probability Analysis}
\label{subsec:phase_type_cov}
We first consider the characterization of coverage probability when $
	P_y \sim {\rm PH}(-\mathbf{S}, \mathbf{p})$, for general $\mathbf{S} \in \R^{n \times n}$.
Using the properties of phase-type distributions and Corollary \ref{cor:gen_campbells}, we may express the coverage probability with respect to $\sinr$ as an expectation of the conditional matrix Laplace transform of $I$. We summarize this in the following corollary.
\begin{cor} (Coverage Probability with Phase-Type Distributed Desired Signal Power)
\\
\label{cor:sinr_cov_ph}
Let $\sinr$ be of the form specified in (\ref{eq:sinr_mod}) and let the matrix Laplace transform of $H(x;y)$ be denoted as $\mtxLaplace_H(\cdot; x,y)$. Then, when $P_y$ follows a PH$(-\mathbf{S}, \mathbf{p})$ distribution the coverage probability corresponding to this SINR model may be expressed as
\begin{align}
	&P_{\rm cov}(\tau) = \nonumber 
	\\
	&\mathbf{p}^{\rm T} \E_{L_y}\left[  \exp\left( -\int_{\R^d \setminus B(y, \ell^{-1}(L_y))}\left( \mathbf{I} - \boldsymbol{\mathcal{L}}_H\left(\frac{\tau \ell(\lvert \lvert x - y \rvert \rvert)}{L_y}\mathbf{S}; x,y\right) \right) \Lambda(dx) \right) \exp \left( -\frac{\tau W}{L_y}\mathbf{S} \right)\right] \mathbf{1}_{n \times 1}.
	\label{eq:gen_coverage}
\end{align}
\end{cor}
\begin{proof} This result follows from the CDF of a phase-type distribution stated in (\ref{eq:phase_type_cdf}) in light of the fact that $\mathbf{S}$ is feasible with respect to Corollary \ref{cor:gen_campbells}:
\begin{align}
	&P_{\rm cov}(\tau) = \Pb(\sinr(y) \ge \tau) \nonumber\\
	&= \E\left[ \Pb\left(P_y L_y \ge \tau\left(I(y) + W \right) \big \vert \Phi, I(y)\right) \right] \nonumber
	\\
	&=\E\left[ \mathbf{p}^{\rm T} \exp \left( -\frac{\tau}{L_y}\mathbf{S}(I(y) + W) \right) \mathbf{1}_{n \times 1} \right] \nonumber
	\\
	&= \mathbf{p}^{\rm T} \E\left[  \E\left[\exp \left( -\frac{\tau}{L_y}\mathbf{S}I(y) \right) \bigg \vert L_y\right] \exp \left( -\frac{\tau W}{L_y}\mathbf{S} \right)\right] \mathbf{1}_{n \times 1} \nonumber
	\\
	&= \mathbf{p}^{\rm T} \E_{L_y}\left[  \exp\left( -\int_{\R^d \setminus B(y, \ell^{-1}(L_y))}\left( \mathbf{I} - \boldsymbol{\mathcal{L}}_H\left(\frac{\tau \ell(\lvert \lvert x - y \rvert \rvert)}{L_y}\mathbf{S}; x,y\right) \right) \Lambda(dx) \right) \exp \left( -\frac{\tau W}{L_y}\mathbf{S} \right)\right] \mathbf{1}_{n \times 1}.
\end{align}
The last line follows from Corollary \ref{cor:gen_campbells} and the fact that, conditioned with $L_y$, $\Phi^{!X_y}$ is a PPP with intensity measure 
\begin{align}
	\Lambda^{!X_y}(B) = \Lambda(B \setminus B(y, \ell^{-1}(L_y)) && B \in \mathcal{B}(\R^d).
	\label{eq:int_int_meas}
\end{align}
\end{proof}
Interestingly, in light of the differential representation of matrix functions in Definition \ref{def:mtx_function}, the expression for coverage probability in Corollary \ref{cor:sinr_cov_ph} is closely related to commonly employed expressions in prior work based on Fa{\`a} di Bruno's formula, stated in (\ref{eq:FDB}). However, Corollary \ref{cor:sinr_cov_ph} is arguably more tractable as it alleviates the need to explicitly consider (\ref{eq:FDB}) as it is implicitly entailed in evaluation of the matrix Laplace transform.
 This expression is additionally more general than prior work considering coverage probability with phase-type fading, e.g \cite{Alfano11}.
 
  Moreover, this formulation allows for one to determine coverage probability through the characterization of the matrix integral, 
\begin{align}
	\int_{\R^d \setminus B(y, \ell^{-1}(L_y))}\left( \mathbf{I} - \boldsymbol{\mathcal{L}}_H\left(\frac{\tau \ell(\lvert \lvert x - y \rvert \rvert)}{L_y}\mathbf{S}; x,y\right) \right) \Lambda(dx).
	\label{eq:mtx_fad_int}
\end{align}
This is analogous to the case where $P_y \sim {\rm Exp}(\lambda)$, which corresponds to the well known Rayleigh fading scenario, in which case one need only consider the scalar integral
\begin{align}
	\int_{\R^d \setminus B(y, \ell^{-1}(L_y))}\left( 1 - {\mathcal{L}}_H\left(\frac{\tau \ell(\lvert \lvert x - y \rvert \rvert)}{L_y}\lambda; x,y\right) \right) \Lambda(dx).
	\label{eq:scalar_fad_int}
\end{align}
Indeed, in light of Corollary \ref{cor:alt_gen_campbells} we may express (\ref{eq:gen_coverage}) in terms of the higher order derivatives of (\ref{eq:scalar_fad_int}) using the Jordan normal form of $\mathbf{S}$. This opens the door to methods employed to obtain tractability in the analysis of (\ref{eq:scalar_fad_int}) in the setting of Rayleigh fading in prior work to be leveraged in this more general setting.

\subsection{Meta-Distribution Analysis}
\label{subsec:sinr_meta}

We conclude this section by considering the meta-distribution of the SINR in the case where ${P_y \sim {\rm Exp}(\lambda)}$. Using Corollary \ref{cor:sn_dist}, the meta-distribution may be obtained as follows.

\begin{prp}(Meta-Distribution of the SINR with Exponentially Distributed Desired Signal Power)
\label{prp:sinr_meta_dist}
\\
Let $\sinr$ be of the form specified in (\ref{eq:sinr_mod}) and let the Laplace transform of $H(x;y)$ be denoted as $\Laplace_H(\cdot; x,y)$. Then, when $P_y$ follows an Exp$(\lambda)$ distribution, the meta-distribution corresponding to this SINR model may be expressed as follows. 

Let $N \in \N$, and let $\mathbf{Q}^{(N)}$ denote the negation of the sub-generator matrix corresponding to an Erlang$(N,1)$ random variable. For $\zeta \in [0, 1]$ and $N \in \N$, define the matrix
\begin{align}
	&\mathbf{S}^{(N)}(L_y, \zeta) = \nonumber
	\\
	&\int_{\R^d \setminus B(y, \ell^{-1}(L_y))}\left( \mathbf{I} - \exp\left(-\frac{N}{\left(\log(\zeta^{-1}) - \frac{\lambda\tau W}{L_y} \right)_+}\log\left({\mathcal{L}}_H\left(\frac{\lambda \tau \ell(\lvert \lvert x - y \rvert \rvert)}{L_y}; x,y\right)^{-1}\right) \mathbf{Q}^{(N)} \right) \right) \Lambda(dx).
\end{align}
Then
\begin{align}
&\bar{F}_{P_{\rm S}}(\zeta) = \lim_{N \to \infty} \mathbf{e}_1^{\rm T} \E_{L_y}\left[\exp\left( -\mathbf{S}^{(N)}(L_y, \zeta)\right)\right]\mathbf{1}_{N \times 1}.
\label{eq:sinr_meta_expr}
\end{align}
\end{prp}	

\begin{proof}
	Consider first the conditional success probability. Under the assumption the $P_y \sim {\rm Exp}(\lambda)$ we have
	\begin{align}
	&P_{\rm S}(\tau) = \Pb(\sinr(y) \ge \tau \vert \Phi) \nonumber
	\\
	&= \Pb\left(P_y \ge \frac{\tau}{L_y} (I(y) + W) \bigg\vert \Phi \right) \nonumber
	\\
	&= \E\left[\exp\left(-\frac{\lambda \tau}{L_y} I(y)\right) \bigg\vert \Phi \right] \exp\left(-\frac{\lambda \tau W}{L_y}\right) \nonumber
	\\
	&= \exp\left(-\int_{\R^d \setminus B(y, \ell^{-1}(L_y))}\log\left(\Laplace_H\left(\frac{\lambda \tau \ell(\lvert \lvert x - y \rvert \rvert)}{L_y};x,y \right)^{-1} \right)\Phi^{!X_y}(dx)\right) \exp\left(-\frac{\lambda \tau W}{L_y}\right),
	\end{align}
	where the last line follows from Corollary \ref{thm:campbells}.
Thus, defining the non-negative shot noise 
\begin{align}
	\hat{I}(y; L_y) = \int_{\R^d \setminus B(y, \ell^{-1}(L_y))}\log\left(\Laplace_H\left(\frac{\lambda \tau \ell(\lvert \lvert x - y \rvert \rvert)}{L_y};x,y \right)^{-1} \right)\Phi^{!X_y}(dx),
\end{align}
we have
\begin{align}
	\bar{F}_{P_{\rm S}}(\zeta) &= \Pb\left(P_{\rm S}(\tau) \ge \zeta\right) \nonumber
	\\
	&= \Pb\left(\hat{I}(y; L_y) + \frac{\lambda \tau W}{L_y} \le \log(\zeta^{-1})\right) \nonumber
	\\
	&= \E_{L_y}\left[\Pb\left(\hat{I}(y; L_y) \le \left(\log(\zeta^{-1}) - \frac{\lambda \tau W}{L_y}\right)_+ \bigg \vert L_y \right)\right],
\end{align}
where $(z)_+$ denotes $\max(z,0)$.
(\ref{eq:sinr_meta_expr}) then follows from Corollary \ref{cor:sn_dist} noting that, conditioned on $L_y$, $\Phi^{! X_y}$ is a PPP with an intensity measure of the form specified in (\ref{eq:int_int_meas}).
\end{proof}

Note that Proposition \ref{prp:sinr_meta_dist} additionally implies that upper and lower bounds on the meta-distribution of the form specified in Proposition \ref{prp:approx_bounds} may be obtained using (\ref{eq:sinr_meta_expr}) for finite values of $N$. Moreover, the proof of Proposition \ref{prp:sinr_meta_dist} only requires that the transmit power of the serving BS, $P_y$, be exponentially distributed. Hence, under that condition, the proposition provides a tractable means of characterizing the meta-distribution of SINR models for a wide                                                                                                                       variety of Poisson networks. This provides an alternative approach to the characterization of the meta-distribution employed in prior work, which typically employs the Gil-Paelez theorem to characterize the meta distribution in terms of the imaginary moments of $P_{\rm S}(\tau)$ \cite{Haenggi16}. Notably, it allow for the meta-distribution analysis of SINR, as opposed to simply SIR, models.

\section{Conclusion}
\label{sec:conclusion}
We have developed a generalization of the Laplace transform of a random variable defined as an integral transform with respect to a matrix exponential. Defining this integral transform as the {\em matrix Laplace transform}, we have established its existence as the matrix function generalization of the typical scalar Laplace Transform. We further developed a generalization of Campbell's theorem for the Laplace functional of a PPP to address the characterization of the matrix Laplace transform of Poisson shot noise. 
 The generalization of Campbell's theorem allows for improved tractability in several applications of interest. In particular, we showed how the generalization of Campbell's theorem may be used to obtain higher order moments of Poisson shot noise in terms of elements of its matrix Laplace transform. We additionally established a method to obtain arbitrarily accurate approximations of and upper and lower bounds on the CCDF of Poisson shot noise in terms of its matrix Laplace transform. The general method we used to obtain these characterizations of the CCDF is further applicable to any non-negative random variable. Additionally, we demonstrated how the generalization of Campbell's theorem enables one to characterize the coverage probability of downlink cellular networks when the SINR model has phase-type distributed desired signal power. The resulting expression for coverage probability is analogous to the canonical case when the fading terms in the SINR model correspond to Rayleigh fading. Finally, we showed that, when the SINR model has exponentially distributed desired signal power, the meta-distribution of the SINR may be obtained in terms of the CCDF of a certain instance of Poisson shot noise. The CCDF of this shot noise may, in turn, be tractably characterized using the methods mentioned earlier.

\appendices
\section{Intermediate Lemmas}
\label{app:int_lemmas}
In this appendix, we provide a series of lemmas which will prove useful for the main results detailed in Sec. \ref{sec:analysis}. In the first two lemmas, we  consider Poisson shot noise as defined in the statement of Theorem \ref{thm:gen_campbells},
\begin{align}
	I = \sum_{X_k \in \Phi}H(X_k),
	\label{eq:app_sn_def}
\end{align}
where $\Phi$ is a PPP on $\R^{d}$ with intensity measure $\Lambda$ and $H: \R^{d} \rightarrow \R_{+}$ is a random field independent of $\Phi$ such that $H$ is independent $\Lambda$-a.s.

First, we establish a necessary condition for $I$ to be finite a.s. If the condition does not hold, then $I$ is infinite a.s.
\begin{lem} (Necessary and Sufficient Condition for Finiteness of Poisson Shot Noise)
\\
\label{lem:finiteness}
 Let $\Laplace_H(\cdot;x)$ denote the Laplace Transform of $H(x)$ and consider the integral $ \int_{\R^d} \left(1 - \Laplace_H(s;x) \right) \Lambda(dx)$, $s \in \R_+$.
If the integral is convergent for all $s \in \R_+$, then $I$ is finite a.s. Otherwise, $I$ is infinite a.s.
\end{lem}
\begin{proof} First, assume that $ \int_{\R^d} \left(1 - \Laplace_H(s;x) \right) \Lambda(dx)$ is convergent for all $s \in \R_+$. The integral may be expressed as
\begin{align}
	\int_{\R^d} \left(1 - \Laplace_H(s;x) \right) \Lambda(dx) &= \int_{\R^d} \left(1 - \E[e^{-sH(x)} \vert \Phi] \right) \Lambda(dx) \nonumber
	\\
	&= \int_{\R^d} \E\left[\left(1 - e^{-sH(x)}  \right)\right] \Lambda(dx) \nonumber
	\\
	&= \E\left[\int_{\R^d} \left(1 - e^{-sH(x)}  \right) \Lambda(dx) \right],
\end{align}
where the second line follows from the fact that $H$ and $\Phi$ are independent, and the last line follows from the Fubini-Tonelli theorem. Thus, finiteness of $ \int_{\R^d} \left(1 - \Laplace_H(s;x) \right) \Lambda(dx)$ implies that $\int_{\R^d} \left(1 - e^{-sH(x)}  \right) \Lambda(dx)$ is convergent a.s. Appealing to the version of Campbell's theorem presented in \cite[Thm. 4.6]{Haenggi12}, it follows that $I$ is finite a.s.

Now, assume that $ \int_{\R^d} \left(1 - \Laplace_H(s;x) \right) \Lambda(dx)$ diverges for some $s \in \R_+$, say at $\theta \in \R_+$. Then, for all $\tau \in \R_+$, we have
\begin{align}
	\Pb(I \le \tau) &= \Pb\left(e^{-\theta I} \ge e^{-\theta \tau}\right) \nonumber
	\\
	&\le e^{\theta \tau} \Laplace_I(\theta) \nonumber
	\\
	&= e^{\theta \tau} \exp\left(- \int_{\R^d} \left(1 - \Laplace_H(\theta;x) \right) \Lambda(dx)\right) \nonumber
	\\
	&= 0,
\end{align}
where the second line follows from Markov's Inequality, and the third line follows from the variation of Campbell's theorem presented in Corollary \ref{thm:campbells}. Thus it follows that $I$ is infinite a.s.
\end{proof}

Next, we establish an intermediate form of a generalization of Campbell's theorem which enables the matrix Laplace transform $I$ to be obtained when $\Phi$ is finite a.s.
\begin{lem} (Generalization of Campbell's theorem for Shot Noise Induced by Finite PPPs).
\\
\label{lem:finite_gen_campbells}
Let $\Phi$ be such that $\Lambda(\R^d) < \infty$, and consider the shot noise, $I$, induced by $\Phi$ as in (\ref{eq:app_sn_def}). Let $\mtxLaplace_H(\cdot; x)$ denote the matrix Laplace transform for $H(x)$. Let $\mathbf{S} \in \C^{n \times n}$ be a matrix  such that $\int_{\R^d} \mtxLaplace_H(\mathbf{S};x) \Lambda(dx)$
is convergent. Then, the matrix Laplace transform of $I$ exists at $\mathbf{S}$ and may be expressed as
\begin{align}
	\mtxLaplace_I(\mathbf{S}) = \exp\left(-\int_{\R^d} \left( \mathbf{I} - \mtxLaplace_H(\mathbf{S};x)\right) \Lambda(dx)\right).
\end{align}
\end{lem}
\begin{proof}
	The condition that $\int_{\R^d} \mtxLaplace_H(\mathbf{S};x) \Lambda(dx)$ is convergent implies that $\mtxLaplace_H(\mathbf{S};x)$ exists $\Lambda$-a.s. Using this fact we have
	\begin{align}
		\mtxLaplace_I(\mathbf{S}) &= \E\left[\exp\left( -\sum_{X_k \in \Phi} \mathbf{S} H(X_k) \right) \right] \nonumber
		\\
		&\stackrel{(a)}{=} \E\left[\prod_{X_k \in \Phi} \exp\left( - \mathbf{S} H(X_k) \right) \right] \nonumber
		\\
		&\stackrel{(b)}{=} \E\left[\prod_{X_k \in \Phi}\E\left[ \exp\left( - \mathbf{S} H(X_k) \right) \bigg\vert \Phi\right]\right] \nonumber
		\\
		&\stackrel{(c)}{=} \E\left[\prod_{X_k \in \Phi}\mtxLaplace_H(\mathbf{S}; X_k)\right] \nonumber
		\\
		&\stackrel{(d)}{=} \E\left[\prod_{X_k \in \Phi}\E\left[\mtxLaplace_H(\mathbf{S}; X_k)\big\vert \Phi(\R^d)\right]\right] \nonumber
		\\
		&\stackrel{(e)}{=} \E\left[\left(\frac{1}{\Lambda(\R^d)} \int_{\R^d}\mtxLaplace_H(\mathbf{S}; x) \Lambda(dx) \right)^{\Phi(\R^d)}\right] \nonumber
		\\
		&\stackrel{(f)}{=} e^{-\Lambda(\R^d)}\sum_{k = 0}^{\infty} \frac{\Lambda(\R^d)^k}{k!}\left( \frac{1}{\Lambda(\R^d)} \int_{\R^d} \boldsymbol{\mathcal{L}}_H(\mathbf{S};x) \Lambda(dx)\right)^{k} \nonumber
		\\
		&\stackrel{(g)}{=} \exp\left(-\Lambda(\R^d) \mathbf{I}\right)\exp\left( \int_{\R^d} \boldsymbol{\mathcal{L}}_H(\mathbf{S};x) \Lambda(dx)\right) \nonumber
		\\
		&\stackrel{(h)}{=} \exp\left( -\int_{\R^d}\left( \mathbf{I} - \boldsymbol{\mathcal{L}}_H(\mathbf{S};x) \right) \Lambda(dx) \right).
	\end{align}
	Each step is justified as follows: (a) follows from the fact that $\left(\mathbf{S} H(X_k)\right)_{X_k \in \Phi}$ are commutative, and thus that the matrix exponential of the sum of the sequence is the product of matrix exponentials applied to the sequence \cite[Thm. 10.2]{Higham08}; (b) follows from the fact that $\left(H(X_k)\right)_{X_k \in \Phi}$ are almost surely independent conditioned on $\Phi$; (c) follows from the fact that $\mtxLaplace_H(\mathbf{S};x)$ exists and the definition of the matrix Laplace transform;
	(d) follows from the fact that the atoms of $\Phi$ are independent conditioned on the number of points, $\Phi(\R^d)$; (e) follows from the fact that $\frac{\partial \Pb\left(X_k \in \cdot \vert \Phi(\R^d) \right)}{\partial \Lambda} = \Lambda(\R^d)^{-1}$ in light of the fact that $\Lambda(\R^d)$ is finite; (f) likewise follows from the fact that $\Phi(\R^d)$ is a Poisson random variable due to finiteness of $\Lambda(\R^d)$; (g) follows from the definition of the matrix exponential and the assumption that $\int_{\R^d} \mtxLaplace_H(\mathbf{S};x) \Lambda(dx)$ is convergent; and $(h)$ follows from the fact that $\Lambda(\R^d) \mathbf{I}$ and $\int_{\R^d} \mtxLaplace_H(\mathbf{S};x) \Lambda(dx)$ commutative and that $\Lambda(\R^d)$ is finite.
\end{proof}

We further establish bounds on modulus of the elements of the matrix Laplace transform of non-negative random variables. These are necessary for various steps of the proof of Theorem \ref{thm:gen_campbells} in particular, including the application of Lemma \ref{lem:finite_gen_campbells}.
\begin{lem} (Bounds on the Matrix Laplace Transforms of Non-negative Random Variables)
\\
\label{lem:mtx_bounds}
Let $X$ be a random variable taking values in $\R_+$, and $\mathbf{J} \in \C^{n \times n}$ be a Jordan block matrix with eigenvalue $\lambda \in \C$ such that $\Re(\lambda) > 0$.
	 Additionally, for $j \in \{1, \dots, n\}$ let $Z_j \sim$ Erlang$(j, \Re(\lambda))$ independent of $X$, and let $Z_0 = 0$.
	Then, the matrix Laplace transform of $X$ at $\mathbf{J}$, $\mtxLaplace_X(\mathbf{J})$ admits the following bounds
	\begin{align}
		\left \lvert \mtxLaplace_{X}(\mathbf{J})_{1,j} \right \rvert \le {\Re(\lambda)^{1 - j}} \max\left(\Pb\left(Z_j > X \right), \Pb\left(Z_{j-1} > X \right) \right) \le {\Re(\lambda)^{1 - j}}
		\label{eq:lem1_bd}
	\end{align}
\end{lem}
\begin{proof}
In light of Proposition \ref{prp:mtx_lt} and Definition \ref{def:mtx_function} we have
\begin{align}
	\left \vert \mtxLaplace_{X}(\mathbf{J})_{1,j} \right \rvert 
	&= \frac{1}{(j-1)!}  \left \vert\E\left[ (-X)^{j-1} \exp\left( -\lambda X \right)\right] \right \rvert \nonumber
	\\
	&\le \frac{1}{(j-1)!}  \E\left[ \left \vert (-X)^{j-1} \exp\left( -\lambda X \right) \right \rvert \right] \nonumber
	\\
	&= \frac{1}{(j-1)!}  \E\left[  (X)^{j-1} \exp\left( -\Re\left(\lambda\right) X \right) \right] \nonumber
	\\
	&= \left \vert \mtxLaplace_{X}\left( \Re\left( \mathbf{J} \right)\right)_{1,j} \right \rvert.
\end{align}
The second line follows from Jensen's Inequality. Thus, it follows that $\left \lvert \mtxLaplace_{X}(\mathbf{J})_{1,j} \right \rvert \le \left \lvert \mtxLaplace_{X}(\Re\left(\mathbf{J} \right))_{1,j} \right \rvert$.

Now, define the matrix $\mathbf{V} = {\rm diag}((-\Re(\lambda))^t)_{t = 0}^{n - 1}$ and let $\mathbf{Q} \in \R^{n \times n}$ be an upper triangular Toeplitz matrix with first row
\begin{align}
	\left(\mathbf{Q}\right)_{1,:} = \begin{bmatrix}
		1 &-1 &\mathbf{0}_{1 \times n - 2}
	\end{bmatrix}.
\end{align}
Note that $\Re \left(\mathbf{J} \right)$ may be factored as  $\Re \left(\mathbf{J} \right) = \Re(\lambda)\mathbf{V} \mathbf{Q} \mathbf{V}^{-1}$. Hence, by Proposition \ref{prp:mtx_lt} and \cite[Thm. 1.13]{Higham08} we have
\begin{align}
	\mtxLaplace_{X}\left( \Re\left( \mathbf{J}\right) \right) = \mathbf{V} \mtxLaplace_{X}\left( \Re\left( \lambda \right) \mathbf{Q} \right) \mathbf{V}^{-1}.
\end{align}
Since $\mathbf{V}$ is diagonal it follows that $\mtxLaplace_{X}\left( \Re\left( \lambda \right) \mathbf{Q} \right)$ is an upper triangular Toeplitz matrix. More precisely, we have
\begin{align}
	\mtxLaplace_{X}(\Re\left(\mathbf{J} \right))_{1,j} = (-1)^{j-1}{\Re(\lambda)^{1 - j}}\mtxLaplace_{X}(\Re\left(\lambda \right) \mathbf{Q})_{1,j}
\end{align}
Hence we need only show that 
\begin{align}
\left \lvert \mtxLaplace_{X}(\Re\left(\lambda \right) \mathbf{Q})_{1,j}\right \rvert \le \max\left(\Pb\left(Z_j > X \vert \Phi \right), \Pb\left(Z_{j-1} > X  \right) \right).
\end{align}

To show this, we exploit the fact that $\Re\left(\lambda \right) \mathbf{Q}$ corresponds to the sub-generator matrix of the phase-type representation of an Elang$(n, \Re(\lambda))$ random variable (see Sec. \ref{subsec:phase_type} for a discussion of phase type random variables). In particular, for $j \in \{1, \dots, n\}$  define $\mathbf{w}_j$ as
$\mathbf{w}_j = \begin{bmatrix}
		\mathbf{1}_{1 \times j} &\mathbf{0}_{1 \times n - j}.
	\end{bmatrix}^{\rm T}$. Then, 
we may express $\mtxLaplace_{X}(\Re\left(\lambda \right) \mathbf{Q})_{1,j}$ as
\begin{align}
	\mtxLaplace_{X}(\Re\left(\lambda \right) \mathbf{Q} )_{1, j} = \mathbf{e}_1^{\rm T}\mtxLaplace_{X}(\Re\left(\lambda \right) \mathbf{Q} )\mathbf{w}_{j} - \mathbf{e}_1^{\rm T}\mtxLaplace_{X}(\Re\left(\lambda \right) \mathbf{Q} )\mathbf{w}_{j -1},
\end{align}
For all $j \in \{1, \dots, n\}$ we have
\begin{align}
	\mathbf{e}_1^{\rm T}\mtxLaplace_{X}(\Re\left(\lambda \right) \mathbf{Q})\mathbf{w}_{j} &=  \E\left[\mathbf{e}_1^{\rm T} \exp\left( - \Re\left(\lambda \right) \mathbf{Q} X  \right) \mathbf{w}_{j} \right]   \nonumber
	\\
	&= \Pb\left(\tilde{Z}_j > X \right).
\end{align}
Hence (\ref{eq:lem1_bd}) follows noting that $\left \lvert \Pb\left(Z_j \ge X \right) - \Pb\left(Z_{j-1} \ge X \right) \right \rvert \le \max\left(\Pb\left(Z_j > X \right), \Pb\left(Z_{j-1} > X \right) \right) $.
\end{proof}

\section{Proof of Theorem \ref{thm:gen_campbells}}
\label{app:gen_campbells}

The proof of Theorem \ref{thm:gen_campbells} relies on Lemma \ref{lem:finiteness}, Lemma \ref{lem:finite_gen_campbells}, and Lemma \ref{lem:mtx_bounds}, provided in Appendix \ref{app:int_lemmas}. Lemma \ref{lem:finite_gen_campbells} in particular provides an expression for the matrix Laplace transform of Poisson shot noise when $\Phi$ is finite a.s., and is used in the proof of each part.

Consequently, we will consider the convergence properties of the matrix Laplace transform of the shot noise induced by the restriction of $\Phi$ to bounded subsets of $\R^d$ increasing to the entire space. That is, let $(A_m)_{m \in \mathbb{N}}$ be a sequence of measurable, bounded subsets of $\R^{d}$ such that $A_{m} \subset A_{m+1}$ and $A_m \uparrow \R^d$ as $m \rightarrow \infty$. Then, let $I_m$ denote the shot noise corresponding to the restriction of $\Phi$ to $A_m$, $\Phi \vert A_m$. Note that $\Phi \vert A_m$ is necessarily finite a.s., and thus Lemma \ref{lem:finite_gen_campbells} may be applied to $I_m$.

\subsection{Proof of Part i)}

By assumption, $\lambda$ is in the interior of the region of convergence of $\Laplace_H(\cdot; x)$. Hence, Proposition \ref{prp:mtx_lt} applies and $\mtxLaplace_H(\mathbf{J};x)$ exists. Moreover, by Lemma \ref{lem:mtx_bounds} it follows that
\begin{align}
	\left \lvert \int_{A_m}\mtxLaplace_H(\mathbf{J};x)_{1,j}\Lambda(dx) \right \rvert &\le  \int_{A_m}\left \lvert \mtxLaplace_H(\mathbf{J};x)_{1,j} \right \rvert\Lambda(dx) \nonumber
	\\
	& \le \Re(\lambda)^{1-j} \Lambda(A_m).
\end{align}
Thus, $\int_{A_m}\mtxLaplace_H(\mathbf{J};x)\Lambda(dx)$ is convergent. Hence, by Lemma \ref{lem:finite_gen_campbells} we have
\begin{align}
	\mtxLaplace_{I_m}(\mathbf{J}) 
	&= \exp\left( -\int_{A_m}\left( \mathbf{I} - \boldsymbol{\mathcal{L}}_H(\mathbf{J};x) \right) \Lambda(dx) \right).
\end{align}

Now, we claim that $\mtxLaplace_I(\mathbf{J})$ and $\mtxLaplace_{I_m}(\mathbf{J})$ may be related by the dominated convergence theorem. To see this, note that for $j \in \{1, \dots, n\}$ by Lemma \ref{lem:mtx_bounds} we have
\begin{align}
	\frac{1}{(j-1)!}\E\left[\left  \lvert I_m^{j-1} e^{-\lambda I_m} \right \rvert\right] =\left \lvert \mtxLaplace_{I_m}(\mathbf{J})_{1,j} \right \rvert \le \Re(\lambda)^{1-j}
	\label{eq:bounded_cond}
\end{align}
Thus, noting that $\lim_{m \to \infty} (-I_m)^{j-1} e^{-\lambda I_m}$ exists a.s., we may apply the dominated convergence theorem element-wise to obtain
\begin{align}
	\mtxLaplace_{I}(\mathbf{J}) &= \E\left[\lim_{m \to \infty}\exp\left(-\mathbf{J} I_m \right) \right]  \nonumber
	\\
	&= \lim_{m \to \infty}\mtxLaplace_{I_m}(\mathbf{J})  \nonumber
	\\
	&= \lim_{m \to \infty} \exp\left( -\int_{A_m}\left( \mathbf{I} - \boldsymbol{\mathcal{L}}_H(\mathbf{J};x) \right) \Lambda(dx) \right).
	\label{eq:thm_2_lim}
\end{align}
Hence, to complete the proof of part i) we need only characterize the limit in (\ref{eq:thm_2_lim}).
To that end, we make use of the following lemma.
\begin{lem}
\label{lem:mtx_lt_int} (Limits of $\mtxLaplace_{I_m}(\mathbf{J})$)
\\
The following hold.
\begin{enumerate}[label=\roman*)]
	\item Assume that $\int_{\R^{d}}\left( 1 - \Laplace_H(\lambda;x) \right) \Lambda(dx)$ is convergent. Then, so is $\int_{\R^{d}}\left( \mathbf{I} - \boldsymbol{\mathcal{L}}_H(\mathbf{J};x) \right) \Lambda(dx)$, and
	\begin{align}
		\mtxLaplace_I(\mathbf{J}) = \exp\left( -\int_{\R^d}\left( \mathbf{I} - \boldsymbol{\mathcal{L}}_H(\mathbf{J};x) \right) \Lambda(dx) \right).
	\end{align}
	\item Otherwise, if $\int_{\R^{d}}\left( 1 - \Laplace_H(\lambda;x) \right) \Lambda(dx)$ diverges, or $I$ is infinite a.s., then
	\begin{align}
		\mtxLaplace_I(\mathbf{J}) = \mathbf{0}_{n \times n}.
	\end{align}
\end{enumerate}

\end{lem}

\begin{proof}
	In general, we may express $\exp\left( -\int_{A_m}\left( \mathbf{I} - \boldsymbol{\mathcal{L}}_H(\mathbf{J};x) \right) \Lambda(dx) \right)$ as
	\begin{align}
	\exp & \left(-\int_{A_m}  \left(\mathbf{I} - \boldsymbol{\mathcal{L}}_H(\mathbf{J};x) \right) \Lambda(dx) \right) \nonumber
	\\
	&= \exp\left( -\int_{A_m} \left( 1 - \Laplace_H(\lambda;x) \right) \Lambda(dx)\right) \exp\left( -\int_{A_m}\left( \mtxLaplace_{H}(\mathbf{\lambda \mathbf{I}};x)  - \boldsymbol{\mathcal{L}}_H(\mathbf{J};x) \right) \Lambda(dx) \right). 
	\\
	&= \exp\left( -\int_{A_m} \left( 1 - \Laplace_H(\lambda;x) \right) \Lambda(dx)\right)  \sum_{k = 0}^{n-1} \frac{1}{k!}\left( -\int_{A_m}\left( \mtxLaplace_{H}(\mathbf{\lambda \mathbf{I}};x)  - \boldsymbol{\mathcal{L}}_H(\mathbf{J};x) \right) \Lambda(dx) \right)^{k}
	\end{align}
where the second line follows from the fact that $\int_{A_m}\left( \mtxLaplace_{H}(\mathbf{\lambda \mathbf{I}};x)  - \boldsymbol{\mathcal{L}}_H(\mathbf{J};x) \right) \Lambda(dx)$ is a nilpotent matrix. 
Thus, we have
\begin{align}
	\mtxLaplace_I(\mathbf{J}) = \lim_{m \to \infty}\exp\left( -\int_{A_m} \left( 1 - \Laplace_H(\lambda;x) \right) \Lambda(dx)\right)  \sum_{k = 0}^{n-1} \frac{1}{k!}\left( -\int_{A_m}\left( \mtxLaplace_{H}(\mathbf{\lambda \mathbf{I}};x)  - \boldsymbol{\mathcal{L}}_H(\mathbf{J};x) \right) \Lambda(dx). \right)^{k}
	\label{eq:alt_lim}
\end{align}
We use the above expression to prove both parts.

We first establish i). In light of the fact that each element of $\mtxLaplace_{I_m}(\mathbf{J})$ is bounded by a constant, we have that the limit of (\ref{eq:alt_lim}) exists and must lie in $\C^{n \times n}$. Moreover, under the condition that $\int_{\R^{d}}\left( 1 - \Laplace_H(\lambda;x) \right) \Lambda(dx)$ is convergent, we have that
\begin{align}
	\lim_{m \to \infty} \exp\left( -\int_{A_m} \left( 1 - \Laplace_H(\lambda;x) \right) \Lambda(dx)\right) \ne 0.
\end{align}
Thus, exploiting the fact that $\int_{A_m}\left( \mtxLaplace_{H}(\mathbf{\lambda \mathbf{I}};x)  - \boldsymbol{\mathcal{L}}_H(\mathbf{J};x) \right) \Lambda(dx)$ is nilpotent, by an inductive argument along the limit of the elements of the first row of (\ref{eq:alt_lim}) we then have that
\begin{align}
	\int_{\R^{d}}\left( \mtxLaplace_{H}(\mathbf{\lambda \mathbf{I}};x) - \boldsymbol{\mathcal{L}}_H(\mathbf{J};x) \right) \Lambda(dx)\in \C^{n \times n}.
\end{align}
Therefore,
\begin{align}
	\mtxLaplace_I(\mathbf{J}) &= \lim_{m \to \infty}\exp\left( -\int_{A_m} \left( 1 - \Laplace_H(\lambda;x) \right) \Lambda(dx)\right)  \sum_{k = 0}^{n-1} \frac{1}{k!}\left( -\int_{A_m}\left( \mtxLaplace_{H}(\mathbf{\lambda \mathbf{I}};x)  - \boldsymbol{\mathcal{L}}_H(\mathbf{J};x) \right) \Lambda(dx). \right)^{k} \nonumber
	\\
	&= \exp\left( -\int_{\R^d} \left( 1 - \Laplace_H(\lambda;x) \right) \Lambda(dx)\right)  \sum_{k = 0}^{n-1} \frac{1}{k!}\left( -\int_{\R^d}\left( \mtxLaplace_{H}(\mathbf{\lambda \mathbf{I}};x)  - \boldsymbol{\mathcal{L}}_H(\mathbf{J};x) \right) \Lambda(dx). \right)^{k} \nonumber
	\\
	&= \exp\left( -\int_{\R^d} \left( \mathbf{I} - \mtxLaplace_H(\mathbf{J};x) \right) \Lambda(dx)\right).  
\end{align}

We now establish part ii). First, consider the case where $I$ is infinite a.s. From Lemma \ref{lem:finiteness}, this occurs if $\int_{\R^{d}}\left( 1 - \Laplace_H(s;x) \right) \Lambda(dx)$ is divergent for some $s \in \R_+$. Thus, we have
\begin{align}
	\lvert \mtxLaplace_I(\mathbf{J})_{1,j} \rvert &= \lim_{m \to \infty}  \left \lvert \mtxLaplace_{I_m}(\mathbf{J})_{1,j} \right \rvert \nonumber
	\\
	&\stackrel{(a)}{\le} \lim_{m \to \infty}  \Re(\lambda)^{j-1}\max\left(\Pb(Z_j > I_m),  \Pb(Z_{j- 1} > I_m)\right) \nonumber
	\\
	&\stackrel{(b)}{=} \Re(\lambda)^{j-1}\max\left(\Pb(Z_j > I),  \Pb(Z_{j- 1} > I)\right) \nonumber
	\\
	&\stackrel{(c)}{\le} \Re(\lambda)^{j-1}\max\left(\E\left[e^{Z_j s } \Laplace_{I}(s)\right],  \E\left[e^{Z_{j - 1} s } \Laplace_{I}(s)\right]\right) \nonumber 
	\\
	&= 0,
\end{align}
where (a) follows from Lemma \ref{lem:mtx_bounds}, (b) follows from the dominated convergence theorem, and (c) follows from a Chernoff bound. Hence, $\mtxLaplace_I(\mathbf{J}) = \mathbf{0}_{n \times n}$.

Assume now that $I$ is finite a.s., but that $\int_{\R^{d}}\left( 1 - \Laplace_H(\lambda;x) \right) \Lambda(dx)$ is divergent. Since $I$ is finite a.s., again from Lemma \ref{lem:finiteness} we must have that $\int_{\R^{d}}\left( 1 - \Laplace_H(\Re(\lambda);x) \right) \Lambda(dx)$ is convergent. Appealing to part i), it then follows that $\int_{\R^d}\left( \mtxLaplace_{H}(\mathbf{\Re(\lambda) \mathbf{I}};x)  - \boldsymbol{\mathcal{L}}_H(\Re(\mathbf{J});x) \right) \Lambda(dx)$ is convergent.
This implies that $\int_{\R^d}\left( \mtxLaplace_{H}(\mathbf{\lambda \mathbf{I}};x)  - \boldsymbol{\mathcal{L}}_H(\mathbf{J};x) \right) \Lambda(dx)$ must be convergent, as, for $j \in \{2, \dots, n\}$,
\begin{align}
\left \lvert  \int_{\R^{d}} \left( \mtxLaplace_{H}(\mathbf{\lambda \mathbf{I}};x)  - \boldsymbol{\mathcal{L}}_H(\mathbf{J};x) \right)_{1,j}  \Lambda(dx) \right \rvert &\le   \int_{\R^{d}} \left \lvert \boldsymbol{\mathcal{L}}_H(\mathbf{J};x)_{1,j}   \right \rvert  \Lambda(dx) \nonumber
\\
& \le \int_{\R^{d}} \E\left[\left \lvert (-H(x))^{j-1} e^{-\lambda H(x)} \right \rvert \right] \Lambda(dx) \nonumber
\\
 &= \int_{\R^{d}}  (-1)^{j-1}\mtxLaplace_H(\Re(\lambda);x)_{1,j} \Lambda(dx) \nonumber
 \\
 &= (-1)^{j}\int_{\R^d}\left( \mtxLaplace_{H}(\mathbf{\Re(\lambda) \mathbf{I}};x)  - \boldsymbol{\mathcal{L}}_H(\Re(\mathbf{J});x) \right)_{1,j} \Lambda(dx).
\end{align}
Therefore, appealing additionally to Corollary \ref{thm:campbells} we must have
\begin{align}
	\mtxLaplace_I(\mathbf{J}) &= \lim_{m \to \infty}\exp\left( -\int_{A_m} \left( 1 - \Laplace_H(\lambda;x) \right) \Lambda(dx)\right)  \sum_{k = 0}^{n-1} \frac{1}{k!}\left( -\int_{A_m}\left( \mtxLaplace_{H}(\mathbf{\lambda \mathbf{I}};x)  - \boldsymbol{\mathcal{L}}_H(\mathbf{J};x) \right) \Lambda(dx). \right)^{k} \nonumber
	\\
	&= \exp\left( -\int_{\R^d} \left( 1 - \Laplace_H(\lambda;x) \right) \Lambda(dx) \right) \exp\left( -\int_{\R^d}\left( \mtxLaplace_{H}(\mathbf{\lambda \mathbf{I}};x)  - \boldsymbol{\mathcal{L}}_H(\mathbf{J};x) \right) \Lambda(dx) \right).
\end{align}
Since $\int_{\R^d} \left( 1 - \Laplace_H(\lambda;x) \right) \Lambda(dx)$ is assumed to be divergent, we must have that $\Re\left(\int_{\R^{d}} \left( 1 - \Laplace_H(\lambda;x) \right) \Lambda(dx) \right)$ is infinite. Thus, $\mtxLaplace_I(\mathbf{J}) = \mathbf{0}_{n \times n}$.
\end{proof}

\subsection{Proof of Part ii)}
Note that convergence of $\int_{\R^d} \E[H(x)^{n-1}] \Lambda(dx)$ implies that $\int_{\R^d} \E[H(x)^{j-1}] \Lambda(dx)$, $j \in \{1, \dots, n\}$, is convergent as well. Therefore, it follows that $\mtxLaplace_H(\mathbf{J}_0; x)$ exists $\Lambda$-a.s., and that $\int_{A_m}\boldsymbol{\mathcal{L}}_H(\mathbf{J_0};x)\Lambda(dx)$ is convergent. Hence, we may apply Lemma \ref{lem:finite_gen_campbells} to obtain
\begin{align}
	\mtxLaplace_{I_m}(\mathbf{J}_0) 
	&= \exp\left( -\int_{A_m}\left( \mathbf{I} - \boldsymbol{\mathcal{L}}_H(\mathbf{J_0};x) \right) \Lambda(dx) \right).
\end{align}
We complete the proof of this part noting that  $\lim_{m \to \infty} \mtxLaplace_{I_m}(\mathbf{J}_0) = \mtxLaplace_I(\mathbf{J}_0)$. To see this, note that for $j \in \{1, \dots, n\}$, $\left(\exp\left(-\mathbf{J}_0 I_m\right)\right)_{1,j}$ may be expressed as
\begin{align}
	\left(\exp\left(-\mathbf{J}_0 I_m\right)\right)_{1,j} = \frac{(-1)^j}{(j-1)!} \left(\int_{A_m} H(x) \Phi(dx)\right)^{j-1}.
\end{align}
The right hand side is either non-decreasing or non-increasing in $m$ in light of non-negativity of $H$ and the construction of $(A_m)_{m \in \N}$. Thus, applying the monotone convergence theorem we have
\begin{align}
	\mtxLaplace_I(\mathbf{J}_0) &= \E\left[ \lim_{m \to \infty} \exp\left(-\mathbf{J}_0 I_m\right)\right] \nonumber
	\\
	&= \lim_{m \to \infty} \mtxLaplace_{I_m}\left(\mathbf{J}_0 \right) \nonumber
	\\
	&= \exp\left(-\int_{\R^d}\left( \mathbf{I} - \boldsymbol{\mathcal{L}}_H(\mathbf{J_0};x) \right) \Lambda(dx) \right).
\end{align}

Finally, note that when $\int_{R^d} \E[H(x)^{n-1}] \Lambda(dx) = \infty$ we have
\begin{align}
	\E[I^{n-1}] &= \E\left[\left(\sum_{X_k \in \Phi} H(X_k) \right)^{n-1} \right] \nonumber
	\\
	&\ge \E\left[\sum_{X_k \in \Phi} H(X_k)^{n-1} \right] \nonumber
	\\
	&= \int_{R^d} \E[H(x)^{n-1}] \Lambda(dx) = \infty 
\end{align}
where the last line follows from the Campbell's averaging formula. Hence, $\mtxLaplace(\mathbf{J}_0)$ does not exist.

\subsection{Proof of Part iii)}
By Lemma \ref{lem:finiteness}, the condition that $\int_{\R^d}\left(1 - \Laplace_H(s;x) \right) \Lambda(dx)$ is convergent for all $s \in \R_{+}$ ensures that $I$ is finite a.s. Hence, for $\lambda \in \C$ such that $\Re(\lambda) = 0$, $\Laplace_I(\lambda)$ corresponds to the characteristic function of $I$ and exists and is non-zero for all such $\lambda$ \cite[Thm. 4.6]{Haenggi12}.

We first establish the claim for $I_m$ as defined in the proof of part i). As with part ii), we have that the convergence of $\int_{\R^d} \E[H(x)^{n-1}] \Lambda(dx)$ implies that $\int_{A_m}\boldsymbol{\mathcal{L}}_H(\mathbf{J};x) \Lambda(dx) $ is convergent as well. Thus, we may apply Lemma \ref{lem:finite_gen_campbells} to obtain 
\begin{align}
	\mtxLaplace_{I_m}(\mathbf{J}) 
	&= \exp\left( -\int_{A_m}\left( \mathbf{I} - \boldsymbol{\mathcal{L}}_H(\mathbf{J};x) \right) \Lambda(dx) \right).
\end{align}
Now, note that for $j \in \{1, \dots, n\}$ we have
\begin{align}
	\E\left[ \left \lvert \exp\left(-\mathbf{J} I_m \right)_{1,j} \right \rvert \right] &= \E\left[ \left \lvert \exp\left(-\mathbf{J}_0 I_m \right)_{1,j} \right \rvert \right] \nonumber
	\\
	&\le \E\left[ \left \lvert \exp\left(-\mathbf{J}_0 I \right)_{1,j} \right \rvert \right] \nonumber
	\\
	&= \left \lvert \mtxLaplace_{I}(\mathbf{J}_0)_{1,j} \right \rvert.
\end{align}
Therefore, in light of the fact that $I$ is finite a.s., we may apply the dominated convergence theorem to obtain
\begin{align}
	\mtxLaplace_I(\mathbf{J}) &= \lim_{m \to \infty} \exp\left( -\int_{A_m}\left( \mathbf{I} - \boldsymbol{\mathcal{L}}_H(\mathbf{J};x) \right) \Lambda(dx) \right) \nonumber\\
	&= \lim_{m \to \infty} \exp\left( -\int_{A_m}\left( 1 - \Laplace_H(\lambda;x) \right) \Lambda(dx) \right) \exp\left( -\int_{A_m}\left( \mtxLaplace_H(\lambda \mathbf{I}; x) - \boldsymbol{\mathcal{L}}_H(\mathbf{J};x) \right) \Lambda(dx) \right)\nonumber \\
	&= \exp\left( -\int_{\R^d}\left( 1 - \Laplace_H(\lambda;x) \right) \Lambda(dx) \right) \exp\left( -\int_{\R^d}\left( \mtxLaplace_H(\lambda \mathbf{I}; x) - \boldsymbol{\mathcal{L}}_H(\mathbf{J};x) \right) \Lambda(dx) \right) \nonumber
	\\
	&= \exp\left(-\int_{\R^d}\left( \mathbf{I} - \boldsymbol{\mathcal{L}}_H(\mathbf{J};x) \right) \Lambda(dx) \right).
\end{align}
Note that the second line follows from an identical argument to that presented in Lemma \ref{lem:mtx_lt_int}. The third line follows from Campbell's theorem and the fact that $\int_{\R^d}\left( \mtxLaplace_H(\lambda \mathbf{I}; x) - \boldsymbol{\mathcal{L}}_H(\mathbf{J};x) \right) \Lambda(dx)$ is well defined (a consequence of the fact that $\int_{\R^d} \E[H(x)^{n-1}] \Lambda(dx) < \infty$).

\section{Proof of Proposition \ref{prp:approx_bounds}}
\label{app:approx_bounds_proof}

The proof of i) is a consequence of the dominated convergence theorem. Note that $\tilde{P}^{(N)}_I(\tau)$ may be expressed as
\begin{align}
	\tilde{P}^{(N)}_X(\tau) = \E\left[ \mathds{1}\{ X \ge \tilde{Z}_N \tau\} \right].
\end{align}
Consider the sequence $\chi_N = \mathds{1}\{ X \ge \tilde{Z}_N \tau\}$.
Since $\lim_{N \to \infty} \tilde{Z}_N = 1$ a.s., we have that $\lim_{N \to \infty} \chi_N = \mathds{1}\{ X \ge \tau\}$ a.s. Therefore, as $\lvert \chi_N \rvert \le 1$ for all $N$, we may apply the dominated convergence theorem to conclude that
\begin{align}
	\lim_{N \to \infty} \tilde{P}^{(N)}_X(\tau) = \bar{F}_{X}(\tau).
\end{align}

We now establish ii). The approximation error of $\tilde{P}^{(N)}_X(\tau)$ with respect to $\bar{F}_{X}(\tau)$ may be upper bounded as follows
\begin{align}
	\left \lvert \bar{F}_X(\tau) - \tilde{P}^{(N)}_X(\tau)  \right \rvert &= \left \lvert \E \left[\indicator\{X \ge \tau \} - \Pb(X \ge \tau \tilde{Z}_N \vert X ) \right] \right \rvert \nonumber
	\\
	&=\left \lvert \E \left[\Pb(\tilde{Z}_N > X \tau^{-1} \vert X ) \indicator\{X \ge \tau \} - \Pb(\tilde{Z}_N \le X \tau^{-1} \vert X ) \indicator\{X < \tau \} \right] \right \rvert \nonumber
	\\
	&\le\max\left\{\E \left[\Pb(\tilde{Z}_N > X \tau^{-1} \vert X ) \indicator\{X \ge \tau \}\right], \E\left[\Pb(\tilde{Z}_N \le X \tau^{-1} \vert X ) \indicator\{X < \tau \} \right]\right\}.
	\label{eq:exact_aprx_err}
\end{align}
The upper bound in (\ref{eq:aprx_err_pii}) then follows applying a Chernoff bound to the tail probabilities in (\ref{eq:exact_aprx_err})
\begin{align}
	\Pb(\tilde{Z}_N > X\tau^{-1} \vert X) \le \inf_{\theta \ge 0} \mathcal{L}_{\tilde{Z}_N}(-\theta)e^{-\theta X\tau^{-1}},
\end{align}
and
\begin{align}
	\Pb(Z_N \le X\tau^{-1} \vert X) \le \inf_{\theta \ge 0} \mathcal{L}_{\tilde{Z}_N}(\theta)e^{\theta X\tau^{-1}}.
\end{align}

Finally iii) is established as follows. First, note that (\ref{eq:erlang_aprx_err}) is a corollary of \cite[Thm. 3]{OCinneide91}, which states that $\tilde{Z}_N$ majorizes $Z_N$. Hence,
\begin{align}
	\inf_{\theta \ge 0} \mathcal{L}_{\tilde{Z}_N}(-\theta)e^{-\theta X\tau^{-1}} \ge \inf_{\theta \ge 0} \mathcal{L}_{{Z}_N}(-\theta)e^{-\theta X\tau^{-1}},
\end{align}
and
\begin{align}
	\inf_{\theta \ge 0} \mathcal{L}_{\tilde{Z}_N}(\theta)e^{\theta X\tau^{-1}} \ge \inf_{\theta \ge 0} \mathcal{L}_{{Z}_N}(\theta)e^{\theta X\tau^{-1}}.
\end{align}

The upper and lower bounds in (\ref{eq:aprx_bounds}) may be obtained as follows. First, note that $Z_N$ admits that following Chernoff tail bounds
\begin{align}
	\Pb(Z_N \le 1 + \delta) \le \inf_{\theta \ge 0} \mathcal{L}_{{Z}_N}(-\theta)e^{-\theta (1 + \delta)} =e^{-N(\delta - \log(1+ \delta))},
\end{align}
and
\begin{align}
	\Pb(Z_N \le 1 - \delta) \le \inf_{\theta \ge 0} \mathcal{L}_{{Z}_N}(\theta)e^{\theta (1 - \delta)} =e^{-N(\delta - \log(1+ \delta))}.
\end{align}
Moreover, by definition $\delta(\epsilon)$ statisfies $e^{-N(\delta(\epsilon) - \log(1+ \delta(\epsilon))} = \epsilon$. Thus, it follows that
\begin{align}
	\Pb(Z_N \ge 1 + \delta(\epsilon)) \le \epsilon && \Pb(Z_N \le 1 - \delta(\epsilon)) \le \epsilon.
\end{align}
Using these inequalities we have
\begin{align}
	{P}^{(N)}_X(\tau) &= \E_{Z_N}\left[ \bar{F}_{X}(\tau Z_N)\right] \nonumber
	\\
	&\le \E_{Z_N}\left[ \bar{F}_{X}(\tau Z_N) \vert Z_N \ge 1 - \delta(\epsilon) \right]\Pb(Z_N \ge 1 - \delta(\epsilon)) + \Pb(Z_N \le 1 - \delta(\epsilon)) \nonumber
	\\
	&\le \E_{Z_N}\left[ \bar{F}_{X}(\tau Z_N) \vert Z_N \ge 1 - \delta(\epsilon) \right] + \epsilon \nonumber
	\\
	&\le \bar{F}_{X}(\tau (1 - \delta(\epsilon)))  + \epsilon,
\end{align}
which establishes the lower bound on $\bar{F}_{X}(\tau)$. Similarly, for the upper bound we have
\begin{align}
	{P}^{(N)}_X(\tau) &= \E_{Z_N}\left[ \bar{F}_{X}(\tau Z_N)\right] \nonumber
	\\
	&\ge \E_{Z_N}\left[ \bar{F}_{X}(\tau Z_N) \vert Z_N \le 1 + \delta(\epsilon) \right]\Pb(Z_N \le 1 + \delta(\epsilon)) \nonumber
	\\
	&\ge \E_{Z_N}\left[ \bar{F}_{X}(\tau Z_N) \vert Z_N \le 1 + \delta(\epsilon) \right]\left(1 - \epsilon\right) \nonumber
	\\
	&\ge \bar{F}_{X}(\tau (1 + \delta(\epsilon))) \left(1 - \epsilon\right).
\end{align}

\bibliographystyle{IEEEtran}
\bibliography{../Latex/SpecCoexRefs.bib}
	
\end{document}